\documentclass[12pt]{article}

\usepackage{hyperref}

\usepackage[normalem]{ulem}
\usepackage{amsthm}
\usepackage{amsmath}
\usepackage{amssymb}
\usepackage{mathrsfs}
\usepackage{mathtools}
\usepackage{graphicx}
\usepackage{tikz}
\usetikzlibrary{arrows}
\usepackage{enumerate}
\usepackage{comment}

\usepackage{fullpage}

\usepackage{appendix}

\newtheorem{theorem}{Theorem}

\newtheorem{lemma}[theorem]{Lemma}

\newtheorem{corollary}[theorem]{Corollary}
\newtheorem{fact}[theorem]{Fact}
\theoremstyle{definition}
 
\theoremstyle{remark}

\theoremstyle{remark}

\numberwithin{theorem}{section}

\providecommand{\R}{}
\providecommand{\Z}{}
\providecommand{\N}{}

\renewcommand{\R}{\mathbb{R}}
\renewcommand{\Z}{\mathbb{Z}}
\renewcommand{\N}{{\mathbb N}}

\newcommand{\E}[1]{{\mathbb E}\left[#1\right]}

\newcommand{\p}[1]{\mathbb{P}\left(#1\right)}

\newcommand{\I}[1]{{\mathbf 1}_{[#1]}}

\newcommand{\Cexp}[2]{\mathbb{E}\left[ \left. #1 \; \right| \; #2 \right]}

\newcommand{\abcdDist}{\ensuremath{\mathcal{A}}}
\newcommand{\abcdFour}{\ensuremath{\mathcal{A}^{(4)}}}
\newcommand{\abcdDistFull}{\ensuremath{\mathcal{A}(n,\gamma,\delta,\zeta,\beta,s,\tau,\xi)}}
\newcommand{\abcdLower}{\ensuremath{\mathcal{A}^-(z)}}
\newcommand{\abcdUpper}{\ensuremath{\mathcal{A}^+(z)}}

\newcommand{\round}[1]{\ensuremath{\left\lfloor #1 \right\rceil}}

\newcommand{\tpl}[3]{\ensuremath{\mathcal{P}\left(#1,#2,#3\right)}}

\newcommand{\moment}[2]{\ensuremath{\mu_{#1}(#2)}}

\begin{document}

\title{Self-similarity of Communities of the ABCD Model}

\author{
Jordan Barrett\thanks{Department of Mathematics, Toronto Metropolitan University, Toronto, ON, Canada; e-mail: \texttt{jordan.barrett@torontomu.ca}}, 
Bogumi\l{} Kami\'{n}ski\thanks{Decision Analysis and Support Unit, SGH Warsaw School of Economics, Warsaw, Poland; e-mail: \texttt{bkamins@sgh.waw.pl}}, 
Pawe\l{} Pra\l{}at\thanks{Department of Mathematics, Toronto Metropolitan University, Toronto, ON, Canada; e-mail: \texttt{pralat@torontomu.ca}}, 
Fran\c{c}ois Th\'{e}berge\thanks{Tutte Institute for Mathematics and Computing, Ottawa, ON, Canada; email: \texttt{theberge@ieee.org}}
}

\maketitle

\begin{abstract}
The \textbf{A}rtificial \textbf{B}enchmark for \textbf{C}ommunity \textbf{D}etection (\textbf{ABCD}) graph is a random graph model with community structure and power-law distribution for both degrees and community sizes. The model generates graphs similar to the well-known \textbf{LFR} model but it is faster and can be investigated analytically.

In this paper, we show that the \textbf{ABCD} model exhibits some interesting self-similar behaviour, namely, the degree distribution of ground-truth communities is asymptotically the same as the degree distribution of the whole graph (appropriately normalized based on their sizes). As a result, we can not only estimate the number of edges induced by each community but also the number of self-loops and multi-edges generated during the process. Understanding these quantities is important as (a) rewiring self-loops and multi-edges to keep the graph simple is an expensive part of the algorithm, and (b) every rewiring causes the underlying configuration models to deviate slightly from uniform simple graphs on their corresponding degree sequences.
\end{abstract}

\section{Introduction}\label{sec:intro} 

One of the most important features of real-world networks is their community structure, as it reveals the internal organization of nodes~\cite{fortunato2010community}. In social networks communities may represent groups by interest, in citation networks they correspond to related papers, in the Web graph communities are formed by pages on related topics, etc. Identifying communities in a network is therefore valuable as this information helps us to better understand the network structure.

Unfortunately, there are very few datasets with ground-truth communities identified and labelled. As a result, there is need for synthetic random graph models with community structure that resemble real-world networks to benchmark and tune clustering algorithms that are unsupervised by nature. The \textbf{LFR} (\textbf{L}ancichinetti, \textbf{F}ortunato, \textbf{R}adicchi) model~\cite{lancichinetti2008benchmark,lancichinetti2009benchmarks} is a highly popular model that generates networks with communities and, at the same time, allows for heterogeneity in the distributions of both node degrees and of community sizes. It became a standard and extensively used method for generating artificial networks. 

A similar synthetic network to \textbf{LFR}, the \textbf{A}rtificial \textbf{B}enchmark for \textbf{C}ommunity \textbf{D}etection (\textbf{ABCD})~\cite{kaminski2021artificial} was recently introduced and implemented\footnote{\url{https://github.com/bkamins/ABCDGraphGenerator.jl/}}, including a fast implementation\footnote{\url{https://github.com/tolcz/ABCDeGraphGenerator.jl/}} that uses multiple threads (\textbf{ABCDe})~\cite{kaminski2022abcde}. Undirected variants of \textbf{LFR} and \textbf{ABCD} produce graphs with comparable properties but \textbf{ABCD}/\textbf{ABCDe} is faster than \textbf{LFR} and can be easily tuned to allow the user to make a smooth transition between the two extremes: pure (disjoint) communities and random graph with no community structure. Moreover, it is easier to analyze theoretically---for example, in~\cite{kaminski2022modularity} various theoretical asymptotic properties of the \textbf{ABCD} model are investigated including the modularity function that, despite some known issues such as the ``resolution limit'' reported in~\cite{fortunato2007resolution}, is an important graph property of networks in the context of community detection. Finally, the building blocks in the model are flexible and may be adjusted to satisfy different needs. Indeed, the original \textbf{ABCD} model was recently adjusted to include potential outliers (\textbf{ABCD+o})~\cite{kaminski2023artificial} and extended to hypergraphs (\textbf{h-ABCD})~\cite{kaminski2023hypergraph}\footnote{\url{https://github.com/bkamins/ABCDHypergraphGenerator.jl}}. In the context of this paper, the most important of the above properties is that the \textbf{ABCD} model allows for theoretical investigation of its properties.

\bigskip

Another important aspect of complex networks is self-similarity and scale invariance which are well-known properties of certain geometric objects such as fractals~\cite{mandelbrot1982fractal}. Scale invariance in the context of complex networks is traditionally restricted to the scale-free property of the distribution of node degrees~\cite{barabasi1999emergence} but also applies to the distributions of community sizes~\cite{guimera2003self,clauset2004finding}, degree-degree distances~\cite{zhou2020power}, and network density~\cite{blagus2012self}. Unfortunately, the definition of ``scale free'' has never reached a single agreement~\cite{broido2019scale,holme2019rare} but many experiments provide a statistical significance of these claims such as the experiment on 32 real-world networks that have a wide coverage of economic, biological, informational, social, and technological domains, with their sizes ranging from hundreds to tens of millions of nodes~\cite{zhou2020power}.

In search for more complete self-similar descriptions, methods related to the fractal dimension are considered that use box counting methods and renormalization~\cite{song2005self,gallos2007review,kim2007fractality}. However, the main issue is that complex networks are still not well defined in a proper geometric sense but one may, for example, introduce the concept of hidden metric spaces to overcome this problem~\cite{serrano2008self}. 

For the context of community structure of complex networks, let us highlight one interesting study of the network of e-mails within a real organization that revealed the emergence of self-similar properties of communities~\cite{guimera2003self}. Such experiments suggest that there is some universal mechanism that controls the formation and dynamics of complex networks. 

\bigskip

In this paper, we show that the \textbf{ABCD} model exhibits self-similar behaviour: each ground-truth community inherits power-law degree distribution from the distribution of the entire graph (see Theorem~\ref{thm:main}), that is, the power-law exponent as well as the minimum degree of this distribution are preserved. On the other hand, as in all self-similarities mentioned above, some renormalization needs to be applied. In our case, the distribution is truncated so that the maximum degree, corrected by the noise parameter $\xi$ (see Section~\ref{subsec:the model} for its formal definition), does not exceed the community size.

The above observation, interesting and desired on its own, has some immediate implications that are of interest too. Firstly, we can easily compute the expected volume of each community (see Corollary~\ref{cor:volumes}). Secondly, and more importantly, we can investigate how many self-loops and multi-edges are constructed during the generation process of \textbf{ABCD} (see Theorem~\ref{thm:loops and multiedges}). Understanding this quantity is crucial for two reasons. Firstly, removing these self-loops and multi-edges to obtain a simple graph is a time consuming part of the construction algorithm. Secondly, as the \textbf{ABCD} construction involves several implementations of the well known configuration model, the number of self-loops and multi-edges is directly correlated to how ``skewed'' the final graph is, i.e., more self-loops and multi-edges lead to distributions that are further away from being uniform. We speak about this second reason in more detail in Section~\ref{subsec:abcd construction}.

\bigskip

The paper is structured as follows. In Section~\ref{subsec:the model}, we formally define the \textbf{ABCD} model and state two results that we will need later on: one about the \textbf{ABCD} model and one about the configuration model which is an important ingredient of the \textbf{ABCD} construction process. The main results are presented in Section~\ref{sec:main_result}. Then, in Section~\ref{sec:simulations}, we present results of simulations that highlight properties that are proved in this paper and show their practical implications. Next, the main result (Theorem~\ref{thm:main}) is proved in Section~\ref{sec:coupling}, and its applications (Corollary~\ref{cor:volumes} and Theorem~\ref{thm:loops and multiedges}) are proved in Section~\ref{sec:loops and multiedges}. Finally, a conclusion and some open problems are presented in Section~\ref{sec:conclusions}.

\section{The ABCD Model}\label{subsec:the model}

In this section we introduce the \textbf{ABCD} model. Its full definition, along with more detailed explanations of its parameters and features, can be found in \cite{kaminski2021artificial}. We restate the main components of the \textbf{ABCD} model here to ensure completeness of the exposition in this article.

\subsection{Notation} 

For a given $n \in \N := \{1, 2, \ldots \}$, we use $[n]$ to denote the set consisting of the first $n$ natural numbers, that is, $[n] := \{1, 2, \ldots, n\}$. 

Our results are asymptotic by nature, that is, we will assume that $n \to \infty$. For a sequence of events $(E_n,n \in \N)$, we say $E_n$ holds \emph{with high probability} (\emph{w.h.p.}) if 
\[
\p{E_n} \to 1 \text{ as } n \to \infty \,.
\] 
We say that $E_n$ holds \emph{with extreme probability} (\emph{w.e.p.}) if 
\[
\p{E_n} = 1 - \exp(-\Omega(\log^2 n)) \,.
\] 
In particular, if there are polynomially many events and each holds w.e.p., then w.e.p.\ all of them hold simultaneously.

For random variables $X$ and $Y$ that take on values in $\R$, we say that $X$ is \textit{stochastically bounded from above} by $Y$, and that $Y$ is \textit{stochastically bounded from below} by $X$, if for all $z \in \R$ we have 
\[
\p{X \geq z} \leq \p{Y \geq z} \,.
\]

Power-law distributions will be used to generate both the degree sequence and community sizes so let us formally define it. For given parameters $\gamma \in (0, \infty)$, $\delta, \Delta \in \N$ with $\delta \leq \Delta$, we define a truncated power-law distribution $\tpl{\gamma}{\delta}{\Delta}$ as follows. For $X \sim \tpl{\gamma}{\delta}{\Delta}$ and for $k \in \N$ with $\delta \leq k \leq \Delta$,
\[
\p{X = k} = \frac{\int_k^{k+1} x^{-\gamma} \, dx}{\int_{\delta}^{\Delta+1} x^{-\gamma} \, dx} \,.
\]

Finally, as typical in the field of random graphs, for expressions that clearly have to be an integer, we round up or down but do not specify which: the choice of which does not affect the argument. 

\subsection{The Configuration Model}\label{subsec:config}

The well-known configuration model is an important ingredient of the \textbf{ABCD} generation process so let us formally define it here. Suppose then that our goal is to create a graph on $n$ nodes with a given degree distribution $\textbf{d} := (d_i, i \in [n])$, where $\textbf{d}$ is a sequence of non-negative integers such that $m := \sum_{i \in [n]} d_i$ is even. We define a random multi-graph $\mathrm{CM}(\textbf{d})$ with a given degree sequence known as the \textbf{configuration model} (sometimes called the \textbf{pairing model}), which was first introduced by Bollob\'as~\cite{bollobas1980probabilistic}. (See~\cite{bender1978asymptotic,wormald1984generating,wormald1999models} for related models and results.)


We start by labelling nodes as $[n]$ and, for each $i \in [n]$, endowing node $i$ with $d_i$ half-edges. We then iteratively choose two unpaired half-edges uniformly at random (from the set of pairs of remaining half-edges) and pair them together to form an edge. We iterate until all half-edges have been paired. This process yields $G_n \sim \mathrm{CM}(\textbf{d})$, where $G_n$ is allowed loops and multi-edges and thus $G_n$ is a multi-graph.

\subsection{Parameters of the ABCD Model}

The \textbf{ABCD} model is governed by the following eight parameters.
\[
\begin{array}{|l|l|l|}
\hline
\text{Parameter} & \text{Range} & \text{Description}\\
\hline
n & \N & \text{Number of nodes}\\
\hline
\gamma & (2,3) & \text{Power-law degree distribution with exponent } \gamma\\
\delta & \N & \text{Min degree as least } \delta \\
\zeta & \left( 0, \frac{1}{\gamma-1} \right] & \text{Max degree at most } n^\zeta \\
\hline
\beta & (1,2) & \text{Power-law community size distribution with exponent } \beta\\
s & \N \setminus [\delta] & \text{Min community size at least } s \\
\tau & (\zeta,1) & \text{Max community size at most } n^\tau\\
\hline
\xi & (0,1) & \text{Level of noise}\\
\hline
\end{array}
\]

\subsection{The ABCD Construction}\label{subsec:abcd construction}

We will use $\abcdDist = \abcdDistFull$ for the distribution of graphs generated by the following 5-phase construction process. 

\subsubsection*{Phase 1: creating the degree distribution}
In theory, the degree distribution for an \textbf{ABCD} graph can be any distribution that satisfies (a) a power-law with parameter $\gamma$, (b) a minimum value of at least $\delta$, and (c) a maximum value of at most $n^{\zeta}$. In practice, however, degrees are i.i.d.\ samples from the distribution $\tpl{\gamma}{\delta}{n^\zeta}$. 

For $G_n \sim \abcdDist$, write $\textbf{d}_n = (d_i, i \in [n])$ for the chosen degree sequence of $G_n$ with $d_1 \geq \dots \geq d_n$. Finally, to ensure that $\sum_{i \in [n]} d_i$ is even, we decrease $d_1$ by 1 if necessary; we relabel as needed to ensure that $d_1 \geq d_2 \geq \dots \geq d_n$. This potential change has a negligible effect on the properties we investigate in this paper and we thus only present computations for the case when $d_1$ is unaltered. 

\subsubsection*{Phase 2: creating the communities}
We next assign communities to the \textbf{ABCD} model. Similar to the degree distribution, the distribution of community sizes must satisfy (a) a power-law with parameter $\beta$, (b) a minimum value of $s$, and (c) a maximum value of $n^{\tau}$. In addition, we also require that the sum of community sizes is exactly $n$. Again, we use a more rigid distribution in practice: communities are generated with sizes determined independently by the distribution $\tpl{\beta}{s}{n^\tau}$. We generate communities until their collective size is at least $n$. If the sum of community sizes at this moment is $n + k$ with $k > 0$ then we perform one of two actions: if the last added community has size at least $k+s$, then we reduce its size by $k$. Otherwise (that is, if its size is $c < k+s$), then we delete this community, select $c$ old communities and increase their sizes by 1. This again has a negligible effect on the analysis and we thus only present computations for the case when community sizes are unaltered.

For $G_n \sim \abcdDist$, write $L$ for the (random) number of communities in $G_n$ and write $\textbf{C}_n = (C_j,j \in [L])$ for the chosen collection of communities in $G_n$ with $|C_1| \geq \dots \geq |C_L|$ (again, let us stress the fact that $\textbf{C}_n$ is a random vector of random length $L$).

\subsubsection*{Phase 3: assigning degrees to nodes}
At this point in the construction of $G_n \sim \abcdDist$ we have a degree sequence $\textbf{d}_n$ and a collection of communities $\textbf{C}_n$. Initially, each community $C_j$ contains $|C_j|$ \textit{unassigned} nodes, i.e., nodes that have not been assigned a label or a degree. We then iteratively assign labels and degrees to nodes as follows. Starting with $i=1$, let $U_i$ be the collection of unassigned nodes at step~$i$. At step $i$ choose a node uniformly at random from the set of nodes $u$ in $U_i$ that satisfy 
\[
d_i \leq \frac{|C(u)|-1}{1-\xi \phi} \,,
\]
where $C(u)$ is the community containing $u$ and
\[
\phi = 1 - \frac{1}{n^2} \sum_{j \in [L]} | C_j |^2 \,,
\]
and assign this node label $i$ and degree $d_i$; we have that $U_{i+1} = U_i \setminus \{u\}$. (We bound the degrees assignable to node $u$ in community $C$ to ensure that there are enough nodes in $C \setminus\{u\}$ for $u$ to pair with, preventing guaranteed loops or guaranteed multi-edges during phase 4 of the construction. We expect a $1-\xi \phi$ fraction of the half-edges attached to $u$ to end up in community $C$, hence the choice for the bound. See~\cite{kaminski2022modularity} or~\cite{kaminski2021artificial} for a more detailed explanation.)

\subsubsection*{Phase 4: creating edges}
At this point $G_n$ contains $n$ nodes labelled as $[n]$, partitioned by the communities $\textbf{C}_n$, with node $i \in [n]$ containing $d_i$ unpaired half-edges. The last step is to form the edges in $G_n$. Firstly, for each $i \in [n]$ we split the $d_i$ half-edges of $i$ into two distinct groups which we call \textit{community} half-edges and \textit{background} half-edges. For $a \in \Z$ and $b \in [0,1)$ define the random variable $\round{a+b}$ as
\[
\round{a+b} = \bigg\{
\begin{array}{ll}
a & \text{ with probability } 1-b, \text{ and}\\
a+1 & \text{ with probability } b \,.
\end{array}
\]
Now define $Y_i := \round{(1-\xi) d_i}$ and $Z_i := d_i - Y_i$ (note that $Y_i$ and $Z_i$ are random variables with $\E{Y_i} = (1-\xi) d_i$ and $\E{Z_i} = \xi d_i$) and, for all $i \in [n]$, split the $d_i$ half-edges of $i$ into $Y_i$ community half-edges and $Z_i$ background half-edges. Next, for all $j \in [L]$, construct the \textit{community graph} $G_{n,j}$ as per the configuration model on node set $C_j$ and degree sequence $(Y_i, i \in C_j)$. Finally, construct the \textit{background graph} $G_{n,0}$ as per the configuration model on node set $[n]$ and degree sequence $(Z_i,i \in [n])$. In the event that the sum of degrees in a community is odd, we pick a maximum degree node $i$ in said community and replace $Y_i$ with $Y_i + 1$ and $Z_i$ with $Z_i - 1$. Once again, this minor adjustment has a negligible effect on the analysis and we thus assume that none of these sums are odd. Note that $G_{n,j}$ is a graph and $C_j$ is the set of nodes in this graph; we refer to $C_j$ as a \textit{community} and $G_{n,j}$ as a \textit{community graph}. Note also that $G_n = \bigcup_{0 \leq j \leq n} G_{n,j}$.

\subsubsection*{Phase 5: rewiring collisions}
Note that, although we are calling $G_{n,0},G_{n,1},\dots,G_{n,L}$ \textit{graphs}, they are in fact \textit{multi-graphs} at the end of phase~4. To ensure that $G_n$ is simple, we perform a series of \textit{rewirings} in $G_n$. A rewiring takes two edges as input, splits them into four half-edges, and creates two new edges distinct from the input. We first rewire each community graph $G_{n,j}$ independently as follows.
\begin{enumerate}
\item For each edge $e \in E(G_{n,j})$ that is either a loop or contributes to a multi-edge, we add $e$ to a \textit{recycle} list that is assigned to $G_{n,j}$.
\item We shuffle the \textit{recycle} list and, for each edge $e$ in the list, we choose another edge $e'$ uniformly from $E(G_{n,j}) \setminus \{e\}$ and attempt to rewire these two edges. We save the result only if the rewiring does not lead to any further collisions, otherwise we give up. In either case, we then move to the next edge in the \textit{recycle} list.
\item After we attempt to rewire every edge in the \textit{recycle} list, we check to see if the new \textit{recycle} list is smaller. If yes, we repeat step 2 with the new list. If no, we give up and move all of the ``bad'' edges from the community graph to the background graph. 
\end{enumerate}
We then rewire the background graph $G_{n,0}$ in the same way as the community graphs, with the slight variation that we also add edge $e$ to \textit{recycle} if $e$ forms a multi-edge with an edge in a community graph or, as mentioned previously, if $e$ was moved to the background graph as a result of giving up during the rewiring phase of its community graph. At the end of phase~5, we have a simple graph $G_n \sim \abcdDist$.

\bigskip

Note that phase~5 of the \textbf{ABCD} construction process exists only to ensure that $G_n$ is simple. Thus, if one were satisfied with a multi-graph $G_n$ that had all of the properties $\abcdDist$ offers, one could simply terminate the process after phase~4. However, for most practical uses such as community detection, we require a simple graph and thus require phase~5. As mentioned in Section~\ref{sec:intro}, phase~5 is a time consuming part of the algorithm. Theorem~\ref{thm:loops and multiedges} gives us some insight as to why that is the case, namely, because with high probability the number of self-loops and multi-edges generated during phase~4 is at least $\Omega(L)$. Theorem~\ref{thm:loops and multiedges} is therefore quite valuable as it lets us know when our choice of $\gamma,\beta,\zeta$ and $\tau$ will yield a best-case-scenario number of self-loops and multi-edges (in expectation).

Theorem~\ref{thm:loops and multiedges} is also valuable for helping us understand how ``skewed'' the community graphs, along with the background graph, are with respect to graphs generated uniformly at random from the set of simple graphs on the respective degree sequences. In \cite{janson2020random}, Janson shows that if a graph is constructed as the configuration model on degree sequence $\mathbf{d}$, followed by a series of switchings, then a relatively small number of switchings yields a distribution that is asymptotically equal (with respect to the total variation distance) to the uniform distribution on simple graphs with degree sequence $\mathbf{d}$. By extrapolating this result, we can infer that the number of switchings required in phase~5 of the \textbf{ABCD} construction process is directly correlated with how ``skewed'' the resulting graph is.

\subsection{Known Results for \textbf{ABCD} and Configuration Models}\label{sec:old_results}

A result from~\cite{kaminski2022modularity} that we use often in this paper is a tight bound on the number of communities generated by the \textbf{ABCD} model.

\begin{theorem}[\cite{kaminski2022modularity} Corollary~5.5 (a)]\label{thm:previous_paper}
Let $G_n \sim \abcdDist$ and let $L$ be the number of communities in $G_n$. Then w.e.p.\ 
the number of communities, $L$, is equal to 
\[
L = L(n) = \left( 1 + O\left( (\log n)^{-1} \right) \right) \, \hat{c} n^{1-\tau(2-\beta)} \,,
\]
where
\[
\hat{c} = \frac{2-\beta}{(\beta-1)s^{\beta-1}} \,.
\]
\end{theorem} 
Likewise, a result of the configuration model that we use many times in the proof of Theorem~\ref{thm:loops and multiedges} is the following. 
\begin{theorem}\label{thm:S and M bound}
Let $G_n$ be a graph constructed via the configuration model with degree sequence $\mathbf{d} = (d_i, i \in [z])$, and let $S$ and $M$ be the number of loops and multi-edges in $G_n$. Then
\begin{align}
\E{S} 
&= 
\frac{\sum_{i \in [z]} q_i(q_i-1)}{2\left( \sum_{i \in [z]} q_i - 1 \right)} 
\leq
\frac{1}{2} \frac{\sum_{i \in [z]} q_i^2}{\sum_{i \in [z]} q_i - 1}   \,, \text{ and} \label{eq:S bound} \\
\E{M} 
&\leq
\frac{\sum_{1 \leq i < j \leq z} q_i(q_i-1)q_j(q_j-1)}{2\left( \sum_{i \in [z]} q_i - 1 \right)\left( \sum_{i \in [z]} q_i - 3 \right)} 
\leq
\frac{1}{2}\frac{\sum_{1 \leq i < j \leq z} q_i^2q_j^2}{\left( \sum_{i \in [z]} q_i - 3 \right)^2} \,. \label{eq:M bound} 
\end{align}
\end{theorem}
See Chapter~7 in \cite{van2016random} for a detailed study on the number of loops and multi-edges in a configuration model. In particular, the equality in (\ref{eq:S bound}) and the first inequality in (\ref{eq:M bound}) come from respective equations (7.3.21) and (7.3.26) in \cite{van2016random}.

\section{Main Result}\label{sec:main_result}

Our main result is a stochastic bound on the degree sequence of a given community in $\abcdDist$. For $G_n \sim \abcdDist$ with degree sequence $\textbf{d}_n$, and for community graph $G_{n,j}$ with nodes from $C_j$, we make the following distinction: the \textit{degree sequence of $G_{n,j}$} is the degree sequence of the community graph $G_{n,j}$, whereas the \textit{degree sequence of $C_{j}$} is the subset of $\textbf{d}_n$ containing the degrees of nodes in $C_j$. Hence, the degree sequence of $C_j$ is $(d_v, v \in C_j)$ and the degree sequence of $G_{n,j}$ is $(Y_v, v \in C_j)$ where recall that $Y_v = \round{(1-\xi)d_v}$. 

\begin{theorem}\label{thm:main}
Let $G_n \sim \abcdDist$. Let $C_j$ be a community in $G_n$ with $|C_j| = z$ and let $\textbf{c}_j$ be the degree sequence of community $C_j$. Next, let $\epsilon = \epsilon(n) = n^{-(\tau-\zeta)(2-\beta)/2} = o(1)$, let 
\[
\Delta_z = \min\left\{ \frac{z-1}{1-\xi \phi}, n^\zeta \right\} \,,
\] 
and let $X^-$ and $X^+$ be random variables with the following probability distribution functions on $\{\delta,\dots,\Delta_z\}$. 
\begin{align*}
\p{X^- = k} 
&=
\frac{\int_k^{k+1} x^{-\gamma} \, dx}{\int_\delta^{\Delta_z+1} x^{-\gamma} \, dx} \,, \text{ and } \\
\p{X^+ = k} 
&=
\frac{\left( 1 - \epsilon \I{k = \delta} \right)\int_k^{k+1} x^{-\gamma} \, dx}{(1-\epsilon)\int_{\delta}^{\delta+1} x^{-\gamma} \, dx + \int_{\delta+1}^{\Delta_z+1} x^{-\gamma} \, dx} = (1+o(1)) \p{X^- = k} \,. 
\end{align*}
Finally, let $(X_i^-, 1\leq i \leq z)$ and $(X_i^+,1\leq i \leq z)$ be i.i.d.\ sequences with $X_i^- \sim X^-$ and $X_i^+ \sim X^+$. Then w.h.p.\ $\textbf{c}_j$ is stochastically bounded below by $(X_i^-, 1\leq i \leq z)$ and above by $(X_i^+,1\leq i \leq z)$.
\end{theorem}

The proof of Theorem~\ref{thm:main} is provided in Section~\ref{sec:coupling}. The power of this theorem is that it allows us to compare the structure of community graphs in $G_n \sim \abcdDist$ with the structure of graphs constructed via the configuration model on an i.i.d.\ degree sequence that is well understood. In this paper we provide two uses of this new and powerful tool. The first is a sharpening of Lemma~5.6 in~\cite{kaminski2022modularity}, describing the volumes of communities in $G_n \sim \abcdDist$. For $X \sim \tpl{\gamma}{\delta}{\Delta}$, write  
\begin{equation}\label{eq:first moment}
\moment{\ell}{\gamma,\delta,\Delta} = \E{X^\ell} \,,
\end{equation}
and note in particular that $\moment{1}{\gamma,\delta,n^\zeta}$ is the expected degree of a node in $G_n \sim \abcdDist$. 

\begin{corollary}\label{cor:volumes}
Let $G_n \sim \abcdDist$, let $C_j$ be a community in $G_n$ with $|C_j| = z = z(n)$, and let 
\[
\Delta_z = \min \left\{ \frac{z-1}{1-\xi \phi} , n^\zeta \right\} \,.
\]
Then w.h.p.\
\[
\frac{\E{\mathrm{vol}(C_j)}}{z} = (1+o(1)) \, \moment{1}{\gamma,\delta,\Delta_z}  = 
\begin{cases} 
\left(1+o(1) \right) \moment{1}{\gamma,\delta,n^\zeta} & \text{ if } z(n) \to \infty \,, \text{ and} \\
\Theta\left( \moment{1}{\gamma,\delta,n^\zeta} \right) & \text{ otherwise.}
\end{cases}
\]
\end{corollary}

The second use of Theorem~\ref{thm:main} that we present here is an analysis of the number of loops and multi-edges that are created during phase 4 of the construction process of $G_n \sim \abcdDist$. In practice, phase~5 of the \textbf{ABCD} construction can be computationally expensive. It is therefore valuable to study the number of collisions (loops and multi-edges) generated during phase~4 of the construction. The following theorem tells us that, although w.h.p.\ we can never do better than generating $\Omega(L)$ collisions, where $L$ is the number of communities, we expect to see \textit{at most} $O(L)$ collisions under certain restrictions on $\gamma,\beta,\zeta$, and $\tau$.

\begin{theorem}\label{thm:loops and multiedges}
Let $G_n \sim \abcdDist$ and define the following five variables depending on $G_n$. 
\begin{align*}
S_c &:= \text{ The number of self-loops in community graphs after phase~4.}\\
M_c &:= \text{ The number of multi-edge pairs in community graphs after phase~4.}\\
S_b &:= \text{ The number of self-loops in the background graph after phase~4.}\\
M_b &:= \text{ The number of multi-edge pairs in the background graph after phase~4.}\\
M_{bc} &:= \text{ The number of background edges that are also community edges after phase~4.}
\end{align*}
Then w.h.p.\
\begin{enumerate}
\item $\E{S_c} = O\Big((n^{1-\tau(2-\beta)})(1+n^{\zeta(4-\gamma-\beta))}\Big)$,
\item $\E{M_c} = O\Big((n^{1-\tau(2-\beta)})(1+n^{\zeta(7-2\gamma-\beta))}\Big)$,
\item $\E{S_b} = O(n^{\zeta(3-\gamma)})$,
\item $\E{M_b} = O(n^{\zeta(6-2\gamma)})$, and
\item $\E{M_{bc}} = o(\E{M_c})$.
\end{enumerate}
Moreover, for all valid $\gamma,\beta,\zeta,\tau$, w.h.p.\
\[
\E{S_c} = \Omega(L) \,,
\]
if $\gamma+\beta > 4$ then w.h.p.\
\[
\E{S_c + M_c + M_{bc}} = \Theta(L) \,,
\]
if $2\zeta(3-\gamma) + \tau(2-\beta) \leq 1$ then w.h.p.\
\[
\E{S_b + M_b} = O(L) \,,
\]
and if both inequalities are satisfied then w.h.p.\
\[
\E{S_c + M_c + S_b + M_b + M_{bc}} = \Theta(L) \,.
\]
\end{theorem}

The proofs of Corollary~\ref{cor:volumes} and Theorem~\ref{thm:loops and multiedges}, as well as the surrounding discussions, are presented in Section~\ref{sec:loops and multiedges}.

\section{Simulation Corner}\label{sec:simulations}

Before moving to the proofs of Theorem~\ref{thm:main}, Corollary~\ref{cor:volumes}, and Theorem~\ref{thm:loops and multiedges}, let us present a few experiments highlighting the properties that are proved to hold with high probability. The experiments show that the asymptotic predictions are useful even for graphs on a moderately small number of nodes.

\subsection{The Coupling}

Our main result shows that the degree distribution of a community of size $z$ in $G_n \sim \abcdDist$ is stochastically sandwiched between $(X_i^-, i \in [z])$ and $(X_i^+, i \in [z])$ where $X_i^- \sim \tpl{\gamma}{\delta}{\Delta_z}$ and $X_i^+ \stackrel{d}{\to} X_i^-$ as $n \to \infty$. To compare the degree distribution of communities in \textbf{ABCD} graphs to the stochastic lower-bound $(X_i^-, i \in [z])$, we perform the following experiment. We generate three \textbf{ABCD} graphs $G_n, G_n^*$ and $G_n^{**}$. Consistent in all three graphs are the parameters $n = 2^{20}, \delta = 5, \zeta = 0.4, s = 50, \tau = 0.6$, and $\xi = 0.5$. The graph $G_n$ has unique parameters $\gamma = 2.1$ and $\beta = 1.1$, the graph $G_n^*$ has $\gamma = 2.5$ and $\beta = 1.5$, and $G_n^{**}$ has $\gamma = 2.9$ and $\beta = 1.9$. For each graph, we plot the complementary cumulative distribution function ($\mathrm{ccdf}$) of degrees of (a) the whole graph, (b) the of all smallest communities ($G_n$ had 8 communities of size $s=50$, $G_n^*$ had 29, and $G_n^{**}$ had 82), and (c) the unique largest community (sizes 4074, 4073, and 3903 in respective graphs $G_n,G_n^*$, and $G_n^{**}$). We then plot, in parallel, the expected $\mathrm{ccdf}$s for the three graphs; for the whole graph the $\mathrm{ccdf}$ is that of $\tpl{\gamma}{\delta}{n^\zeta}$, and for the community graphs we use the expected $\mathrm{ccdf}$ of the stochastic lower-bound $(X_i^-, i \in [z])$, i.e., the function $f:\{\delta,\dots,\Delta_z\} \to [0,1]$ where
\[
f(k) 
=
\frac{\int_k^{\Delta_z+1} x^{-\gamma} \, dx}{\int_\delta^{\Delta_z+1} x^{-\gamma} \, dx} 
=
\frac{k^{1-\gamma} - (\Delta_z+1)^{1-\gamma}}{\delta^{1-\gamma} - (\Delta_z+1)^{1-\gamma}} \,.
\]
The results are presented in Figure~\ref{fig:dist}. From these results, we see that the distribution of $(X_i^-, i \in [z])$ is a very good approximation of the distribution of degrees in a community of smallest size as well as a community of largest size. We note that, since $(X_i^-, i \in [z])$ is a lower-bound, we expect the theoretical $\mathrm{ccdf}$ to sit slightly above the empirical $\mathrm{ccdf}$, and this is confirmed by the experiment.

\begin{figure}
\begin{center}
    \includegraphics[scale=0.75]{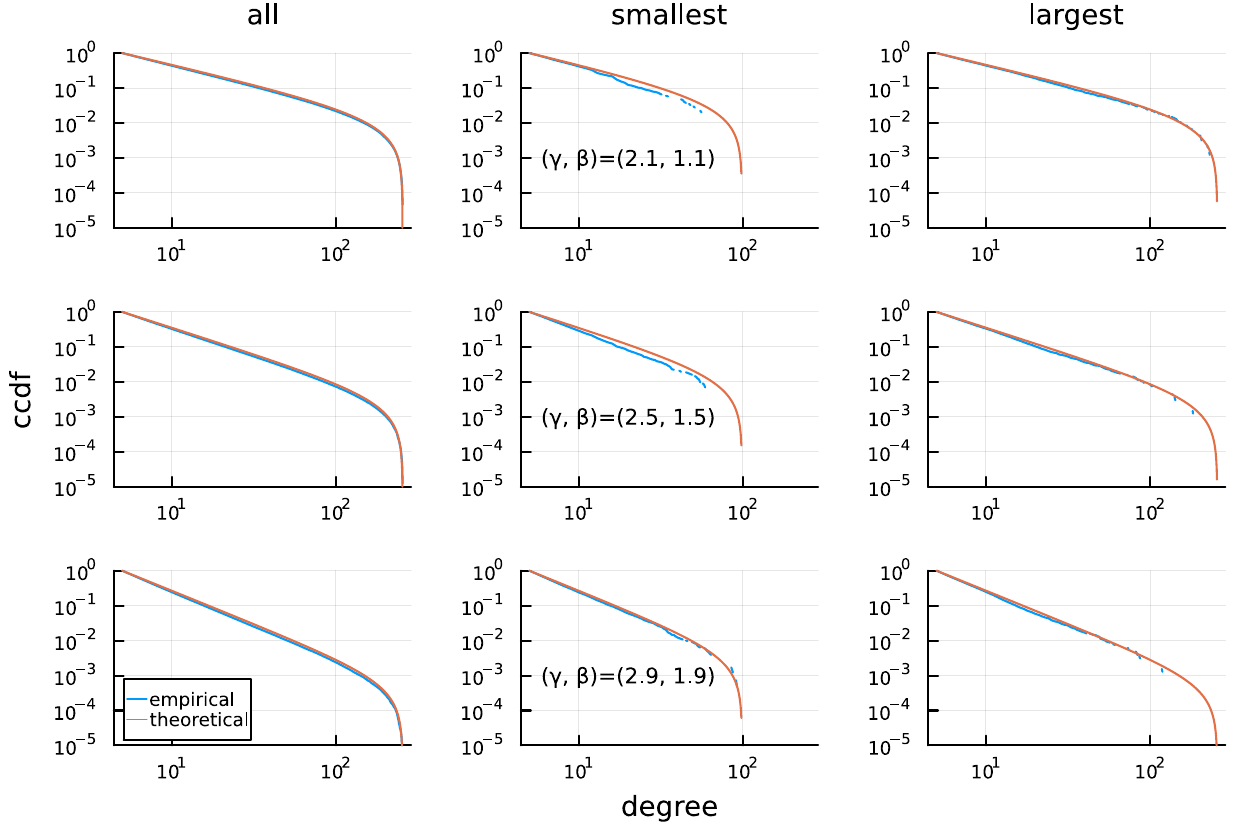}
\end{center}
\caption{The $\mathrm{ccdf}$ for the three different \textbf{ABCD} graphs $G_n$ (top), $G_n^{*}$ (middle), and $G_n^{**}$ (bottom), and for three different subsets of nodes in each graph, namely, the whole graph (left), the union of smallest community graphs (middle), and the unique largest community graph (right). Each function is drawn on a $\log$--$\log$ scale. The blue curves are the empirical data and the orange curves are the theoretical predictions. \label{fig:dist}}
\end{figure}

\subsection{Volumes of Communities}

Next, to investigate how well Corollary~\ref{cor:volumes} predicts the volume of a particular community, we perform the following experiment. We generate three \textbf{ABCD} graphs $G_n, G_n^*$ and $G_n^{**}$. Consistent in all three graphs are the parameters $n = 2^{20}$, $\delta = 5$, $\zeta = 0.6$, $s = 50$, $\tau = 0.9$, and $\xi = 0.5$. The graph $G_n$ has unique parameters $\gamma = 2.1$ and $\beta = 1.1$, the graph $G_n^*$ has $\gamma = 2.5$ and $\beta = 1.5$, and $G_n^{**}$ has $\gamma = 2.9$ and $\beta = 1.9$. In each graph, we sorted communities with respect to their size (from the smallest to the largest) and then grouped them into 10 buckets as equal as possible (that is, the number of communities in any pair of buckets differs by at most one). For each bucket we compute the average degree and the standard deviation over all communities in that bucket. We compare it with the asymptotic prediction based on Corollary~\ref{cor:volumes}, that is, for each community of size $z$ we compute $\moment{1}{\gamma,\delta,\Delta_z}$, and take the average over all communities in the bucket. The results are presented in Figure~\ref{fig:avgdeg2111}. We see that $n = 2^{20}$ is large enough and simulations match the theoretical predictions almost exactly.

\begin{figure}
\begin{center}
\includegraphics[scale=0.33]{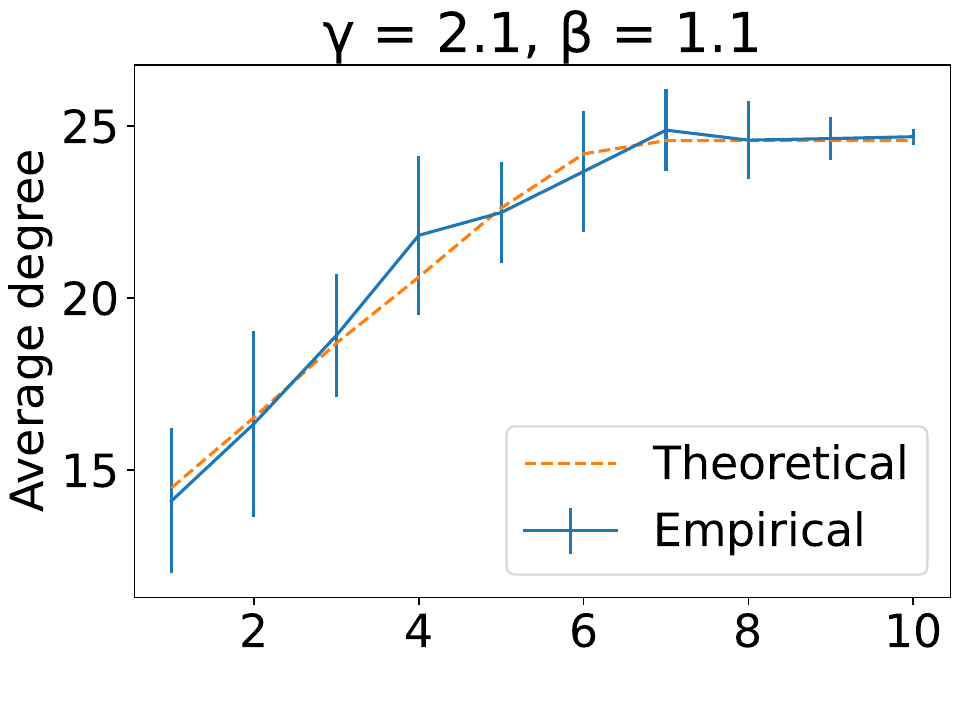}
\includegraphics[scale=0.33]{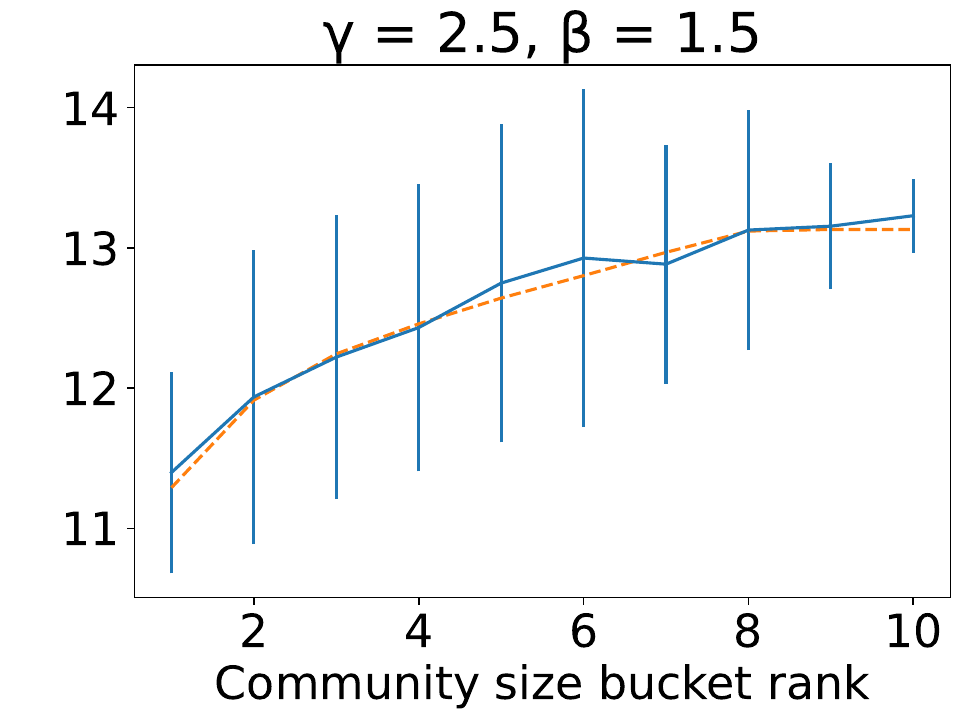}
\includegraphics[scale=0.33]{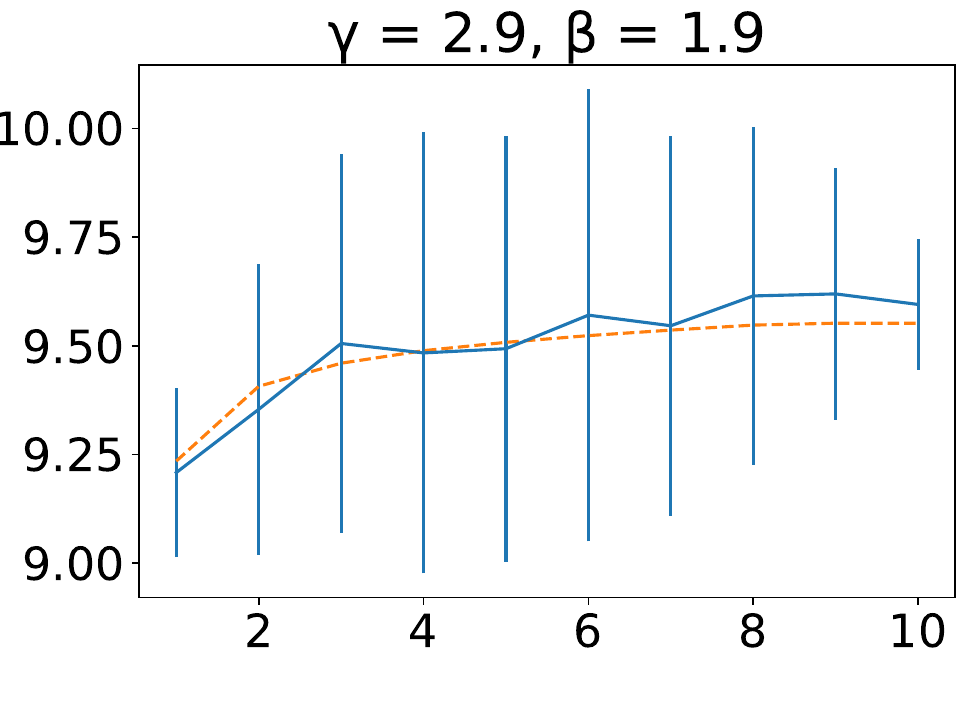}
\end{center}
\caption{The average degrees in communities for $G_n$ (top left), $G_n^*$ (top right), and $G_n^{**}$ (bottom). The communities are ranked by their size and grouped into 10 buckets as equal as possible. The blue line with error bars is the average degree and standard deviation among all communities in each bucket. Note that the errors, in absolute values, are largest for the leftmost plot and smallest for the rightmost plot. The orange dashed line shows the expected volumes for the stochastic lower-bound $(X_i^-,i \in [z])$, computed for each community size and bucketed in the same way as the empirical data. \label{fig:avgdeg2111}}
\end{figure}

\subsection{Loops and Multi-edges}
Finally, to investigate the number of collisions (of various types) generated during phase 4 of the \textbf{ABCD} construction as functions of $n$, we perform the following experiment. For each $n \in  \{2^{15},2^{16},2^{17},2^{18},2^{19},2^{20}\}$, we generate three sequences of 20 \textbf{ABCD} graphs $(G_n(i), i \in [20]), (G_n^*(i), i \in [20])$, and $(G_n^{**}(i), i \in [20])$. Consistent in all three sequences are the parameters $\delta = 5$, $\zeta = 0.6$, $s = 50$, $\tau = 0.9$, and $\xi = 0.5$. The graphs in sequence $(G_n(i), i \in [20])$ have $\gamma = 2.1$ and $\beta = 1.1$, the graphs in $(G_n^*(i), i \in [20])$ have $\gamma = 2.5$ and $\beta = 1.5$, and the graphs in $(G_n^{**}(i), i \in [20])$ have $\gamma = 2.9$ and $\beta = 1.9$. We compare the growth of $S_c/L$, $M_c/L$, $S_b/L$, and $M_b/L$ (the average values and the corresponding standard deviations over 20 graphs), as functions of $n$, for all three sequences. Each sequence represents a different scenario in expectation based on Theorem~\ref{thm:loops and multiedges}, and we comment on each result separately.
\begin{itemize}
\item For $(G_n(i), i \in [20])$ with $\gamma = 2.1$ and $\beta = 1.1$, we have $\gamma + \beta < 4$ and $2\zeta(3-\gamma) + \tau(2-\beta) > \zeta(3-\gamma) + \tau(2-\beta) > 1$ and so we expect each of the variables $S_c/L$, $M_c/L$, $S_b/L$, and $M_b/L$ to be unbounded. In Figure~\ref{fig:collisions2111} we see that, indeed, each of the four variables seem to grow with $n$ in the simulations. 
\item For $(G_n^*(i), i \in [20])$ with $\gamma = 2.5$ and $\beta = 1.5$, we have $\gamma + \beta = 4$ and $ 2\zeta(3-\gamma) + \tau(2-\beta) > 1 > \zeta(3-\gamma) + \tau(2-\beta)$ and so we expect $S_b/L$ to be bounded and $S_c/L,M_c/L,M_b/L$ to be unbounded. As Figure~\ref{fig:collisions2515} shows, the simulations are consistent with the theory for $S_c/L,M_c/L$ and $S_b/L$. However, the trend of $M_b/L$ is unclear. Considering that $2\zeta(3-\gamma) + \tau(2-\beta) = 1.05$ in this case, it is reasonable that the growth of $M_b/L$ should not reveal itself at this scale of $n$.
\item For $(G_n^{**}(i), i \in [20])$ with $\gamma = 2.9$ and $\beta = 1.9$, we have $\gamma + \beta > 4$ and $1 > 2\zeta(3-\gamma) + \tau(2-\beta) > \zeta(3-\gamma) + \tau(2-\beta)$ and so we expect all of $S_c/L,M_c/L,S_b/L,M_b/L$ to be bounded. Figure~\ref{fig:collisions2919} again shows us that theory matches simulations. We note the very slight upward trend of $S_c/L$ and $M_c/L$, likely due to $n$ being too small to see the asymptotic bound take hold.
\end{itemize}

We conclude that Theorem~\ref{thm:loops and multiedges} does a good job at telling us the behaviour of $S_c/L$, $M_c/L$, $S_b/L$, and $M_b/L$ for various $\gamma$ and $\beta$, although the results are not as clear as the other experiments which would likely be resolved by taking larger values of $n$.

\begin{figure}
\[
\begin{array}{cc}
\includegraphics[scale=0.33]{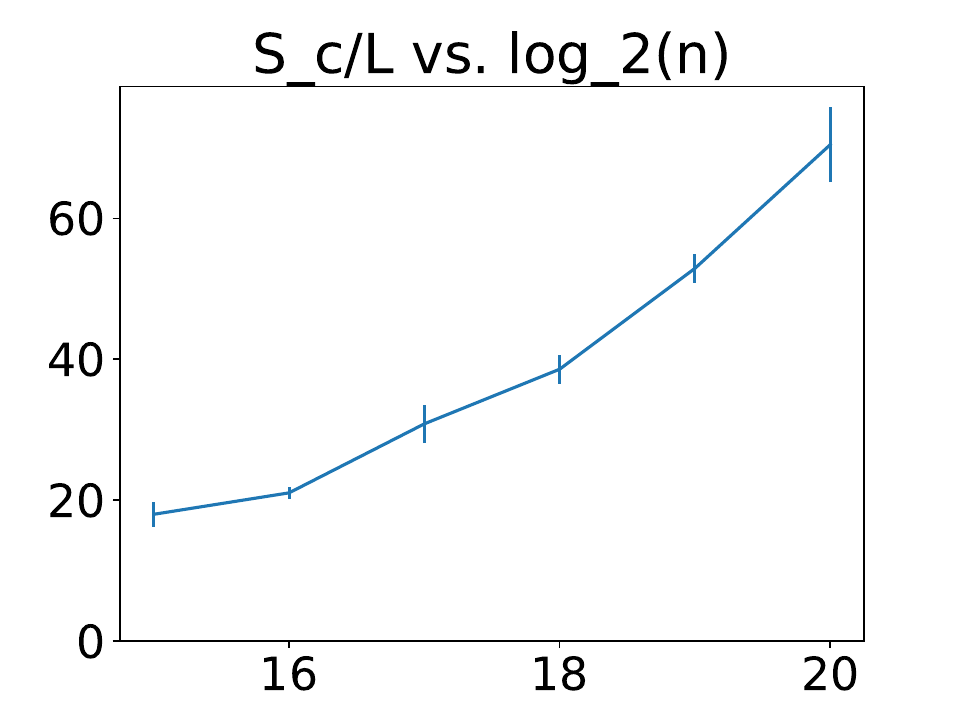}
&
\includegraphics[scale=0.33]{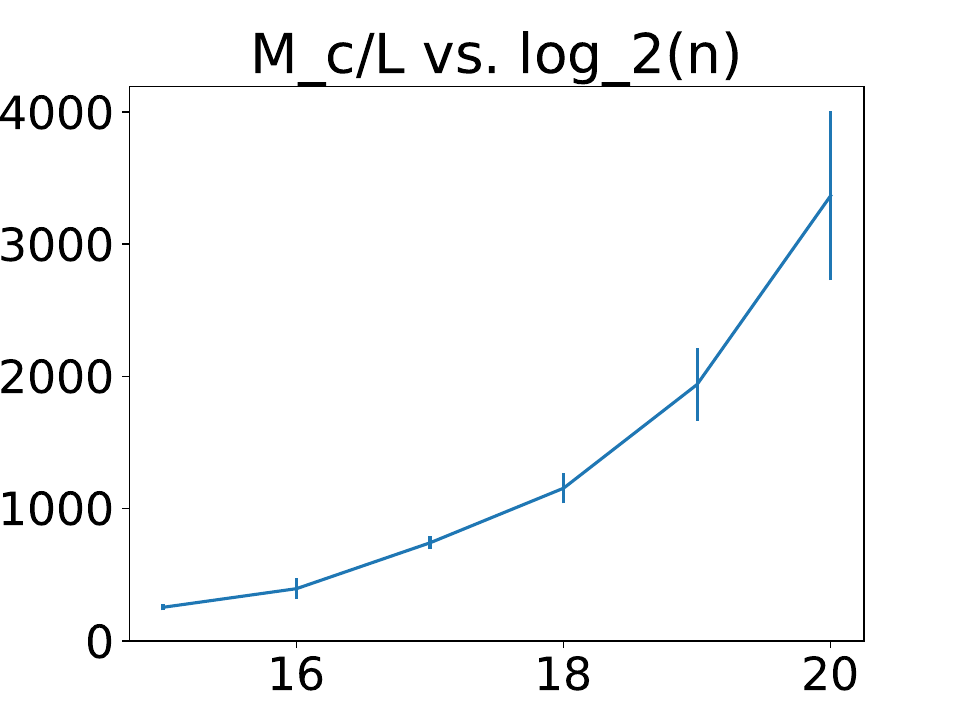}\\
\includegraphics[scale=0.33]{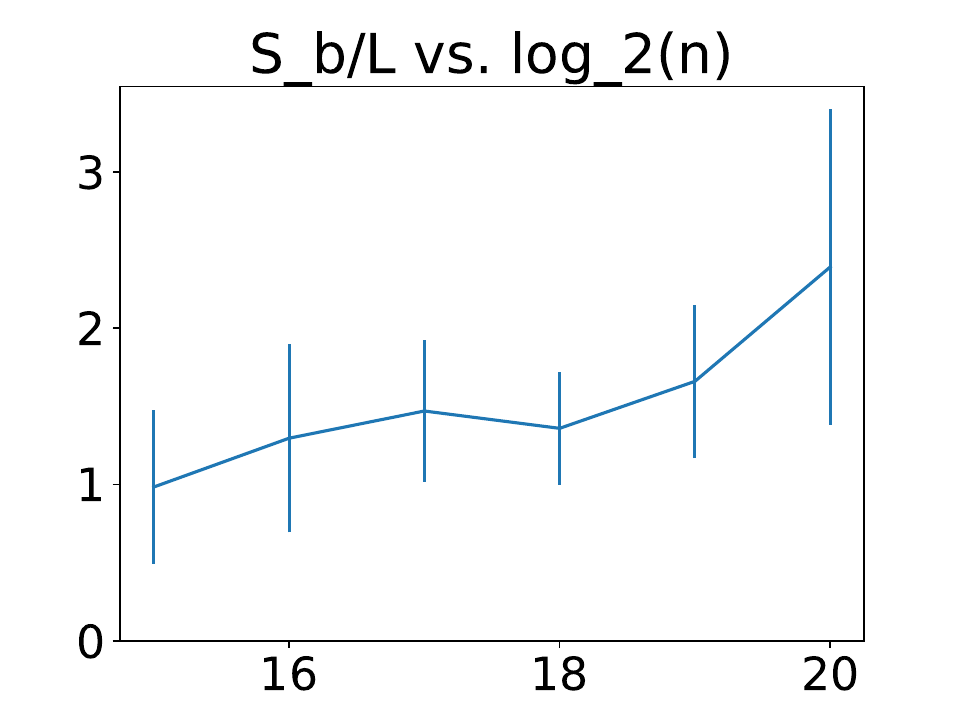}
&
\includegraphics[scale=0.33]{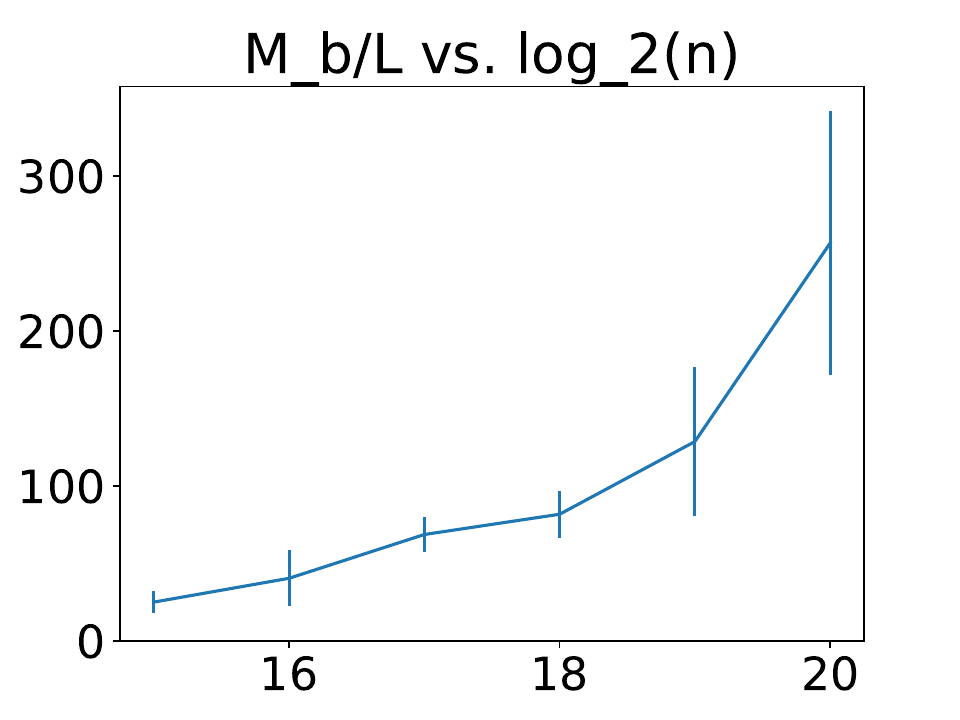}
\end{array}
\]
\caption{\label{fig:collisions2111}In reading order: $S_c/L,M_c/L,S_b/L$ and $M_b/L$ vs.\ $\log_2(n)$ for $(G_n(i), i \in [20])$, averaged over the 20 graphs.}
\end{figure}

\begin{figure}
\[
\begin{array}{cc}
\includegraphics[scale=0.33]{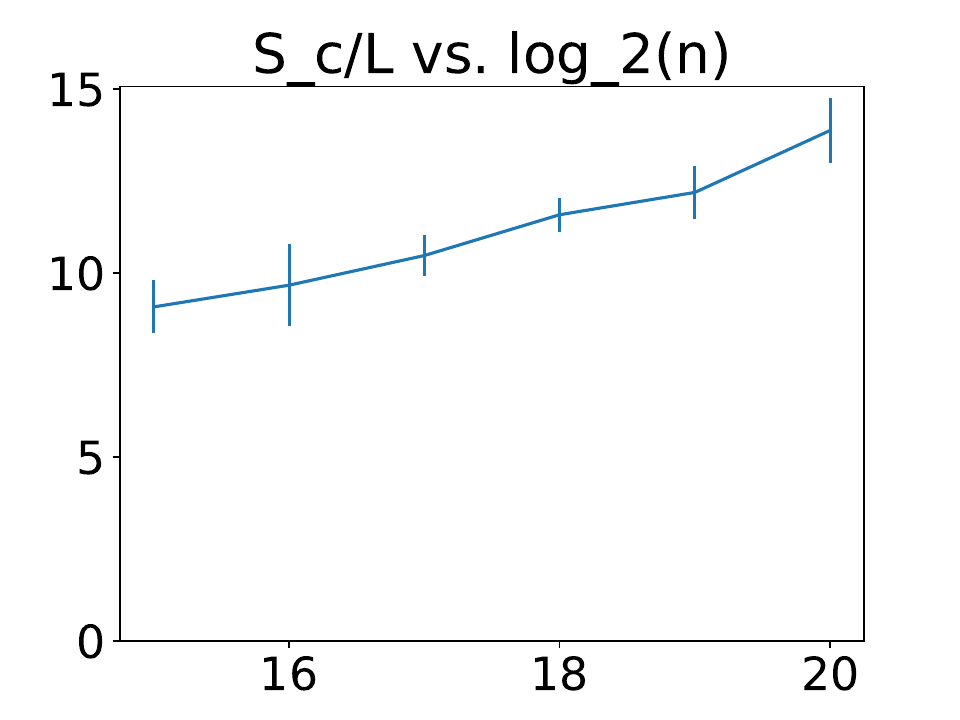}
&
\includegraphics[scale=0.33]{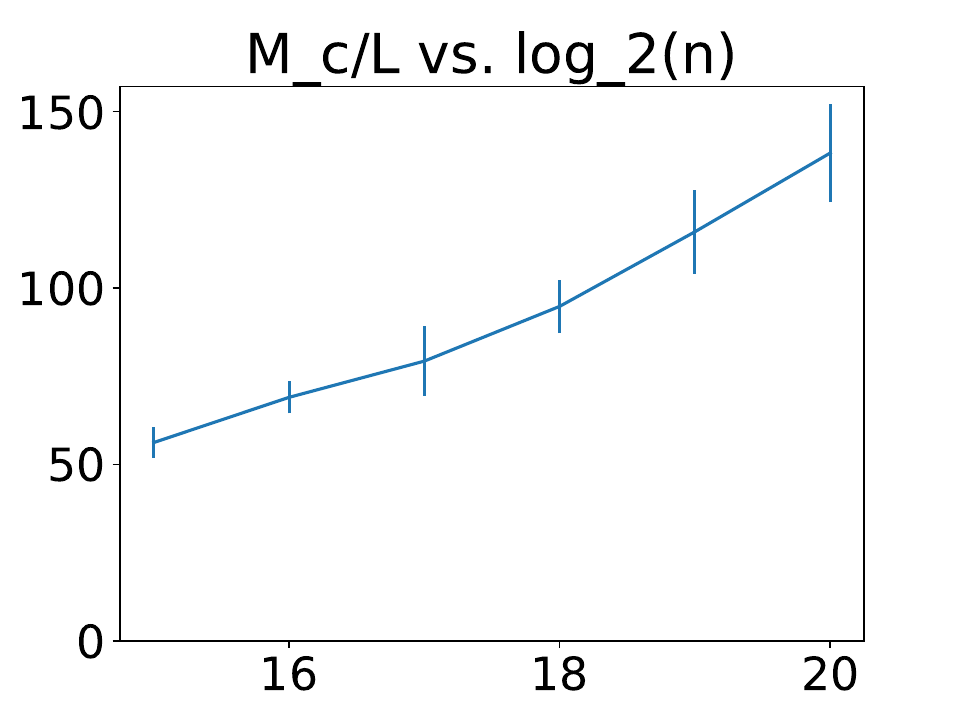}\\
\includegraphics[scale=0.33]{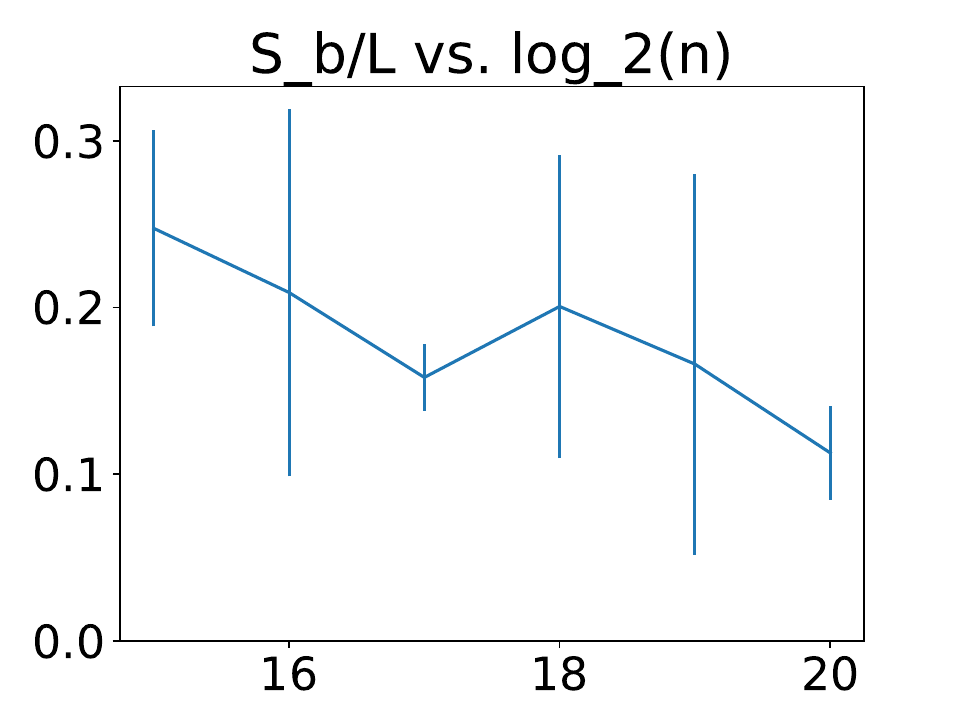}
&
\includegraphics[scale=0.33]{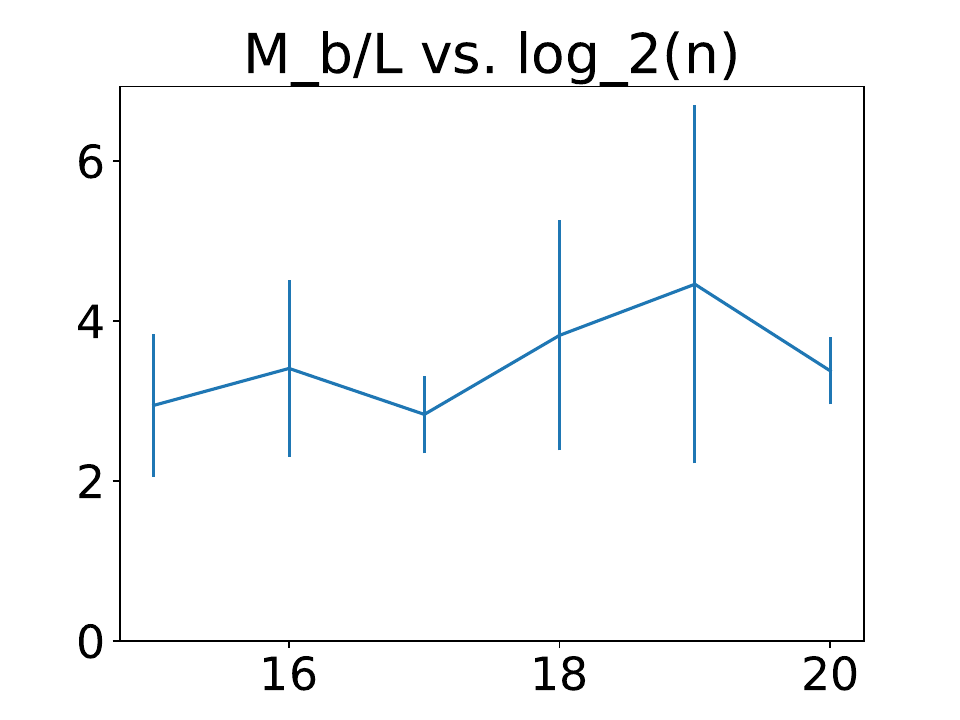}
\end{array}
\]
\caption{\label{fig:collisions2515}In reading order: $S_c/L,M_c/L,S_b/L$ and $M_b/L$ vs.\ $\log_2(n)$ for $(G_n^*(i), i \in [20])$, averaged over the 20 graphs.} 
\end{figure}

\begin{figure}
\[
\begin{array}{cc}
\includegraphics[scale=0.33]{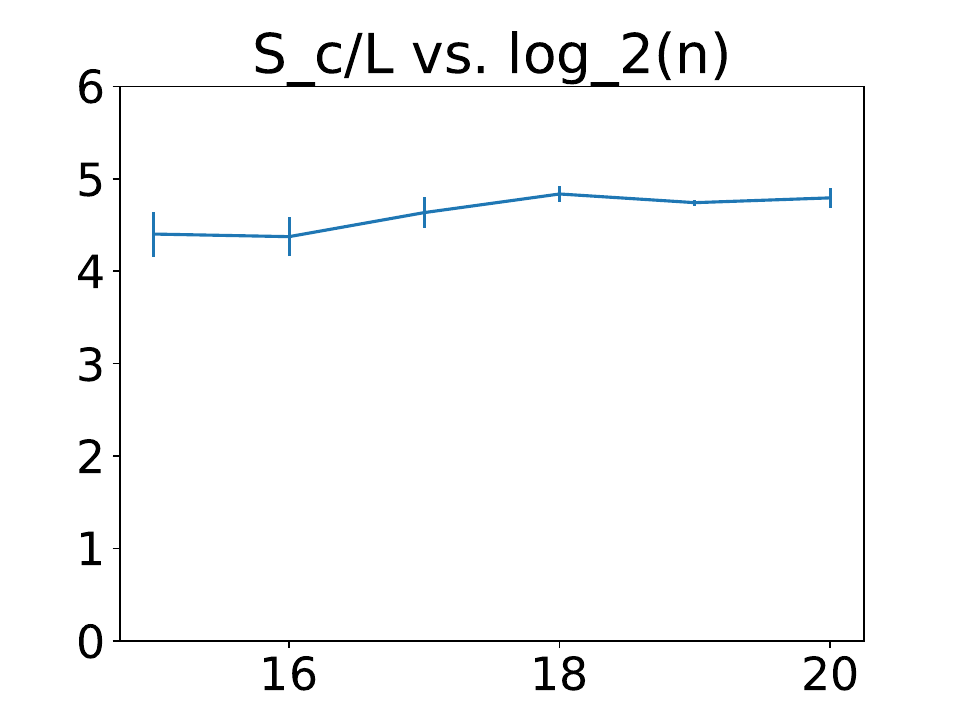}
&
\includegraphics[scale=0.33]{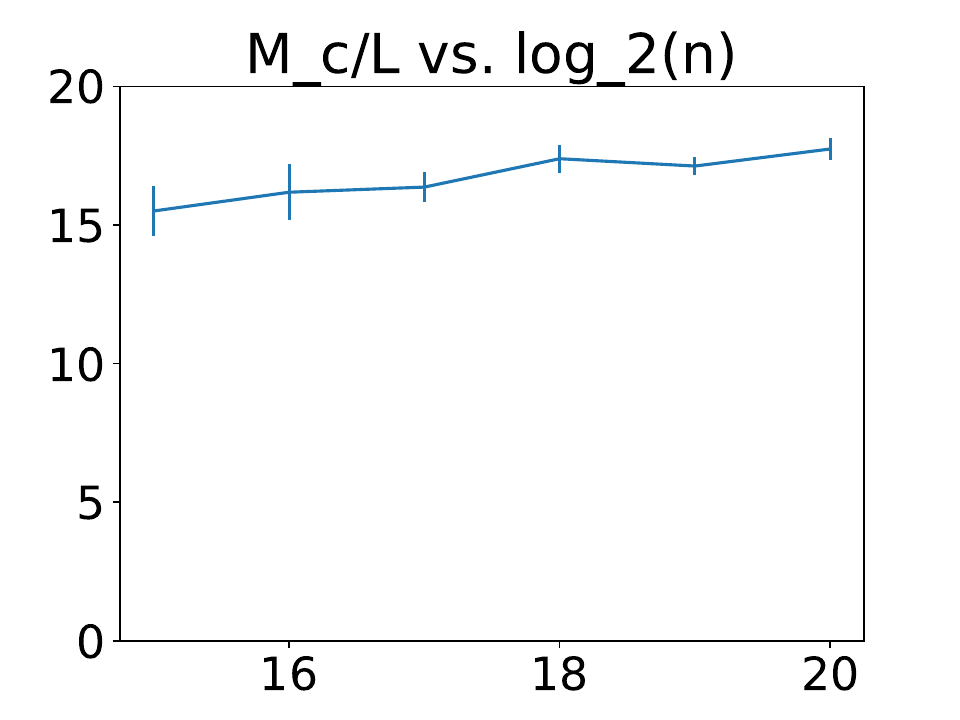}\\
\includegraphics[scale=0.33]{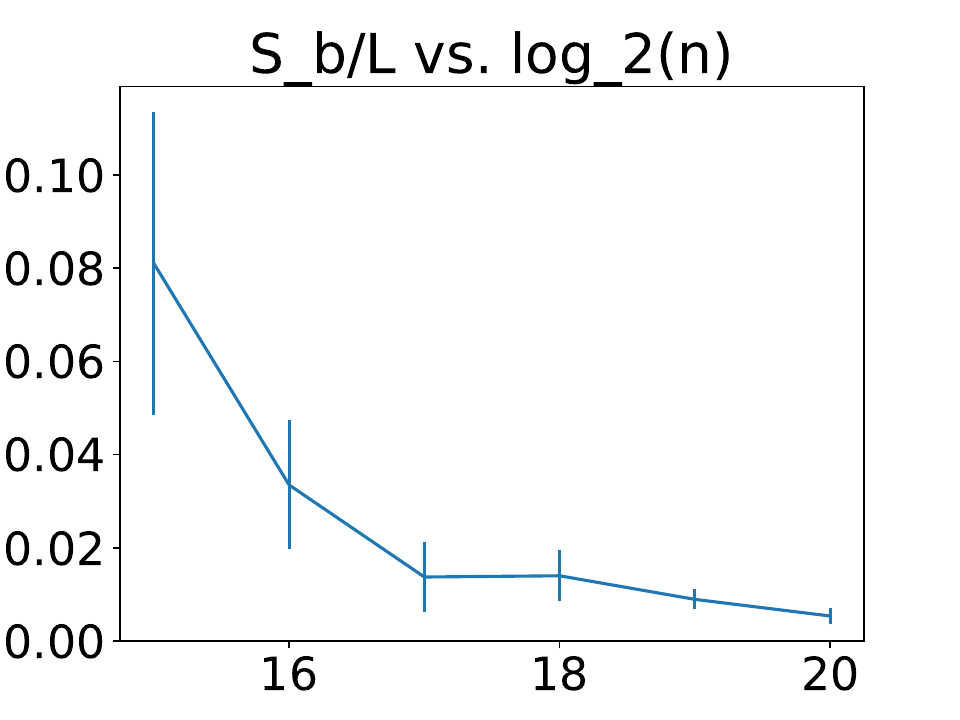}
&
\includegraphics[scale=0.33]{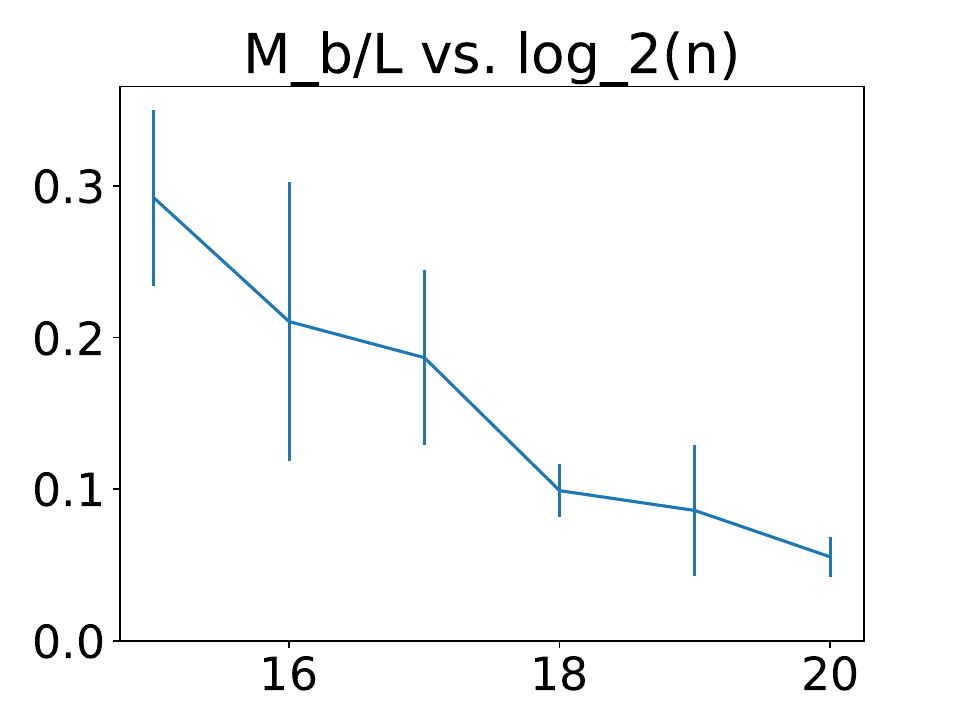}
\end{array}
\]
\caption{\label{fig:collisions2919}In reading order: $S_c/L,M_c/L,S_b/L$ and $M_b/L$ vs.\ $\log_2(n)$ for $(G_n^{**}(i), i \in [20])$, averaged over the 20 graphs.}
\end{figure}

\section{The Coupling (Proof of Theorem~\ref{thm:main})}\label{sec:coupling}

Before we set up a coupling that sandwiches the \textbf{ABCD} construction process in order to control the degree sequence of any community $C_j$, we need to show that almost all nodes belong to large communities. Such communities are large enough such that they can be assigned nodes of any degree. Indeed, since the maximum degree in $G_n$ is (deterministically) at most $n^{\zeta}$, only communities of size less than $n^{\zeta} (1-\xi \phi) + 1 \le n^{\zeta}$ might \emph{not} be available during the entire phase 3 of the \textbf{ACBD} construction process.

\begin{lemma}\label{lem:open communities}
Let $\omega = \omega(n)$ be any function such that $\omega \to \infty$ sufficiently slowly as $n \to \infty$. Next, let $G_n \sim \abcdDist$ and let $V' \subseteq V(G_n)$ be the set of nodes in communities of size at most $n^\zeta$. Then, w.h.p.\ $|V'| < \omega n^{1-(\tau-\zeta)(2-\beta)} = o(n^{1-(\tau-\zeta)(2-\beta)/2}) = o(n)$. 
\end{lemma}

\begin{proof}
Recall that $0 < \zeta < \tau < 1$ and $1 < \beta < 2$. Pick a community $C \in \textbf{C}_n$ uniformly at random and let $X = |C|$ if $|C| \le n^{\zeta}$; otherwise, $X=0$. Then, for $s \leq m \leq n^\zeta$,
\begin{align*}
\p{X = m} 
&= 
\frac{\int_m^{m +1} y^{-\beta} \, dy}{\int_s^{n^{\tau}+1} y^{-\beta} \, dy}\\ 
&= 
(\beta-1) \frac{\int_m^{m +1} y^{-\beta} \, dy}{s^{1-\beta} - (n^\tau+1)^{1-\beta}}\\
&=
\left( 1 + O(n^{\tau(1-\beta)}) \right) (\beta-1) s^{\beta-1} \int_m^{m+1} y^{-\beta} \, dy \,,
\end{align*}
and hence
\begin{align*}
\E{X} 
&= \left( 1 + O(n^{\tau(1-\beta)}) \right) (\beta-1) s^{\beta-1} \sum_{m = s}^{\lfloor n^{\zeta} \rfloor} m \int_m^{m +1} y^{-\beta} dy \\
&\leq
\left( 1 + O(n^{\tau(1-\beta)}) \right) (\beta-1) s^{\beta-1} \int_s^{n^{\zeta}+1} y^{1-\beta} dy \\
&= 
\left( 1 + O(n^{\tau(1-\beta)}) \right) \frac{(\beta-1) s^{\beta-1}}{2-\beta} \left( (n^\zeta+1)^{2-\beta} - s^{2-\beta} \right)\\
&=
\left( 1 + O(n^{\tau(1-\beta)}) + O(n^{\tau(\beta-2)}) \right) \frac{(\beta-1) s^{\beta-1}}{2-\beta} n^{\zeta(2-\beta)} \\
&=
\Theta \left( n^{\zeta(2-\beta)} \right) \,.
\end{align*}
Finally, since w.e.p.\ $L = \Theta\left( n^{1-\tau(2-\beta)} \right)$ (see Theorem~\ref{thm:previous_paper}), we get that
\begin{align*}
\E{V'} 
&=
O\left( n \exp(-\log^2 n) + n^{1-\tau(2-\beta)} \E{X} \right)
= 
O\left( n^{1-(\tau-\zeta)(2-\beta)}  \right) \,,
\end{align*}
and the lemma now follows from Markov's inequality. 
\end{proof}

We will also need the following simple fact about the distribution $\tpl{\gamma}{\delta}{\Delta}$.
\begin{fact}\label{fact:tpl condition}
Fix $\gamma > 0$ and $1 \le \delta \le \delta' \le \Delta' \le \Delta$. Then $X \sim \tpl{\gamma}{\delta}{\Delta}$, conditioned on $\delta' \leq X \leq \Delta'$, has distribution $\tpl{\gamma}{\delta'}{\Delta'}$.
\end{fact}

The remainder of Section~\ref{sec:coupling} is dedicated to proving Theorem~\ref{thm:main}. In the coming arguments and with respect to phase 3 of the \textbf{ABCD} construction process, we refer to a community $C$ as \textit{locked} at step $i$ if $d_i > (|C|-1)/(1-\xi \phi)$ and otherwise we refer to $C$ as \textit{unlocked} at step $i$. We say that a node is locked/unlocked at step $i$ if its corresponding community is locked/unlocked at step $i$. Note that, since $d_1 \le n^{\zeta}$, all communities of size at least $n^{\zeta}(1-\xi \phi)+1$ are always unlocked.

We start with the modified version of phase 3 of the \textbf{ABCD} construction process that will be used to prove the lower bound in Theorem~\ref{thm:main}. Fix $z$ with $s \leq z \leq n^\tau$ and define the construction process $\abcdLower$, yielding a collection of degrees assigned to a collection of communities notated as $G_n^-$, as follows.
\begin{enumerate}
\item Copy phases 1 and 2 of the \textbf{ABCD} construction process to get a degree distribution $\textbf{d}_n = (d_i,i \in [n])$ and a collection of communities $\textbf{C}_n = (C_j,j \in [L])$ each containing unassigned nodes (recall that unassigned nodes are nodes that have not yet been assigned a label or a degree). 
\item Copy phase 3 of the \textbf{ABCD} construction process until the communities of size $z$ are unlocked. This event occurs at step $i$ where $i$ is the smallest label satisfying $d_i \leq \frac{z-1}{1-\xi \phi}$ (recall that the degree sequence $\mathbf{d}_n = (d_i, i \in [n])$ is non-increasing and that label $i$ and degree $d_i$ are assigned to an unassigned node at time $i$). At this point, all communities of size at least $z$ are unlocked and $i-1$ nodes that belong to communities of size at least $z+1$ have been assigned a label and a degree. 
\item Now unlock all communities and assign labels $i,\dots,n$ and corresponding degrees $d_i,\dots,d_n$ to the unlabelled nodes in $[n]$ uniformly at random. 
\end{enumerate}

We will first show that a community $C_j$ in $G_n^-$ of size $z$ has the desired degree distribution. 

\begin{lemma}\label{lem:deg dist lower}
Fix $z=z(n)$ such that $s \leq z \leq n^\tau$. Let $G_n^- \sim \abcdLower$, let $C_j$ be a community in $G_n^-$ with $|C_j| = z$ and with degree sequence $\mathbf{c}_z^-$, and let $(X_i^-, 1 \leq i \leq z)$ be the i.i.d.\ sequence defined in Theorem~\ref{thm:main}. Then, $\mathbf{c}_z^- \stackrel{d}{=} (X_i^-, 1 \leq i \leq z)$. 
\end{lemma}

\begin{proof}
To prove the lemma, we will use the well-known Principle of Deferred Decisions. A~simple but useful observation is that when constructing $G_n^-$ one can defer exposing some information about the degree sequence $\textbf{d}_n$ to the very end. Indeed, during phase 1 of the \textbf{ABCD} construction, we may only expose information whether $d_i \leq \frac{z-1}{1-\xi \phi}$ or not; if $d_i > \frac{z-1}{1-\xi \phi}$, then we expose $d_i$ but otherwise we only reveal that $d_i \leq \frac{z-1}{1-\xi \phi}$. This partial information is enough to continue with the auxiliary process of constructing $G_n^-$.

Recall that community $C_j$ is locked as long as $d_i > \frac{z-1}{1-\xi \phi}$. Let $i$ be the smallest label such that $d_i \leq \frac{z-1}{1-\xi \phi}$. (Note that, in particular, if $n^{\zeta} \le \frac{z-1}{1-\xi \phi}$, then $C_j$ is immediately unlocked, that is $i=1$.) Once we unlock $C_j$ in $G_n^-$ at step $i$, we unlock all communities and assign degrees $d_i,\dots,d_n$ uniformly to the set of unassigned nodes in $[n]$. Thus, $\mathbf{c}_z^-$ is a uniform subsequence of $(d_i,\dots,d_n)$ of size $z$. Now, we finally expose the degrees in this subsequence. By Fact~\ref{fact:tpl condition}, each $d_i$ follows precisely a truncated power law with upper bound $\Delta_z = \min\left\{ \frac{z-1}{1-\xi \phi}, n^\zeta \right\}$ and lower bound $\delta$. Thus, $\mathbf{c}_z^- \stackrel{d}{=} (X_i^-, 1 \leq i \leq z)$, proving the lemma.
\end{proof}

We are now ready to couple the auxiliary process constructing $G_n^-$ with the original process generating $G_n$, the \textbf{ABCD} graph. This will prove the lower bound in Theorem~\ref{thm:main}.

\begin{proof}[Proof of Theorem~\ref{thm:main} (lower bound)]
Construct $G_n^- \sim \abcdLower$ with nodes labelled as $[n]$, degree sequence $\textbf{d}_n = (d_i, i \in [n])$, and community sequence $\textbf{C}_n = (C_j,j \in [L])$. Next, for all $i \in [n]$ define $z_i = \lceil d_i(1-\xi \phi) + 1 \rceil$; note that a community $C$ is unlocked in phase 3 of the \textbf{ABCD} construction at the first step $i$ for which $|C| \geq z_i$). Now construct $G_n$ in parallel with $G_n^-$ as follows. 
\begin{enumerate}
\item Let $G_n$ have degree sequence $\textbf{d}_n$ and community sequence $\textbf{C}_n$. 
\item Copy the degree assignment process of $G_n^-$ until the communities of size $z$ are unlocked. Let $i$ be the smallest label satisfying $d_i \leq \frac{z-1}{1-\xi \phi}$. Instead of unlocking all communities as we do in $G_n^- \sim \abcdLower$, we will unlock only those communities $C$ satisfying 
$$
|C| \geq z_i = \lceil d_i(1-\xi \phi) + 1 \rceil
$$ 
as we do in $\abcdDist$. (Note that, if $|C| = z \geq \lceil n^\zeta(1-\xi \phi) + 1 \rceil$, then $i=1$ and $C$ is unlocked from the start.)
\item Now, for $j \in \{i,\dots ,n\}$ starting with $j = i$, we first unlock all communities $C$ satisfying $|C| \geq z_j$. We then partition the nodes into four sets. We say that node $v$ is \textit{open} in $G_n$ at step $j$ if $v$ is both unlocked and unlabelled before step $j$, and otherwise we say $v$ is closed at step $j$ (and similarly for $G_n^-$). The four sets are as follows:
\begin{align*}
V_j^{++} &= \left\{ v : v \text{ is open in both } G_n^- \text{ and } G_n \text{ at step } j \right\} \,, \\
V_j^{+-} &= \left\{ v : v \text{ is open in } G_n^- \text{ and closed in } G_n \text{ at step } j \right\} \,, \\
V_j^{-+} &= \left\{ v : v \text{ is closed in } G_n^- \text{ and open in } G_n \text{ at step } j \right\} \,, \\
V_j^{--} &= \left\{ v : v \text{ is closed in both } G_n^- \text{ and } G_n \text{ at step } j \right\} \,.
\end{align*}
Note that $V_i^{+-}$ is the set of nodes in communities of size at most $z_i-1$ and $V_i^{-+} = \emptyset$. However, all four sets will change with $j$. We now choose a node $v$ in $G_n^-$ to receive label $j$ and degree $d_j$ as per the $\abcdLower$ construction (note that $v$ is a uniform element of $V_j^{++} \cup V_j^{+-}$). We then choose a node in $G_n$ to receive label $j$ and degree $d_j$ as follows.
\begin{itemize}
\item If $v \in V_j^{++}$, then we give label $j$ and degree $d_j$ to $v$ in $G_n$.
\item If $v \in  V_j^{+-}$, then we give label $j$ and degree $d_j$ to a uniform node in $V_j^{-+}$ with probability $p_j$, and to a uniform node in $V_j^{++}$ with probability $1-p_j$, where
\[
p_j = \frac{|V_j^{++}||V_j^{-+}| + |V_j^{+-}||V_j^{-+}|}{|V_j^{++}||V_j^{+-}| + |V_j^{+-}||V_j^{-+}|} \,;
\]
we will later verify that $p_j \leq 1$. 
\end{itemize}
\item Once all nodes have been assigned a degree, create the community edges and background edges in $G_n$ as per the usual $\abcdDist$ construction process.
\end{enumerate}
We claim (a) that $G_n \sim \abcdDist$, and (b) that any community $C \in \textbf{C}_n$ of size $z$ with $G_n$-degree sequence $\mathbf{c}_z$ and $G_n^-$-degree sequence $\mathbf{c}_z^-$ satisfies $\mathbf{c}_z \geq \mathbf{c}_z^-$ point-wise.  

Starting with claim (a), it is clear by the construction process $\abcdLower$ that $\textbf{d}_n$ and $\textbf{C}_n$ are valid sequences for $G_n \sim \abcdDist$. We must then verify that, for $j  = i,\dots,n$, the node in $G_n$ chosen to receive label $j$ and degree $d_j$ is a uniform node from the set of unlabelled nodes in communities of size at least $d_j(1-\xi \phi) + 1$. Note that this set of nodes is precisely $V_j^{++} \cup V_j^{-+}$, and so we need only show that, for $u,v \in V_j^{++} \cup V_j^{-+}$, the probability of labelling $u$ and the probability of labelling $v$ are equal. We will first show that $p_j \leq 1$ by showing that $|V_j^{-+}| \leq |V_j^{+-}|$ for all $j \in \{i,\dots,n\}$. In fact, we will show a stronger result, namely, that $|V_j^{+-}|-|V_j^{-+}|$ is precisely the number of nodes in communities that are locked in $G_n$ at time~$j$.

As mentioned earlier, when $j = i$, $V_j^{+-}$ is the set of nodes in communities that are still locked (that is, of size at most $z_i-1$) and $V_j^{-+} = \emptyset$, so the desired property holds. Now suppose the property holds up to some time $j \geq i$. At step $j$, if $v \in V_j^{++}$ receives label $j$ and degree $d_j$ in $G_n^-$, then $v$ also receives this label and degree in $G_n$, and thus $v$ is moved from $V_j^{++}$ to $V_{j+1}^{--}$ ($|V_{j+1}^{+-}|-|V_{j+1}^{-+}|$ is unaffected by this event). On the other hand, if $v \in V_j^{+-}$ receives label $j$ and degree $d_j$ at step $j$, then $v$ is moved from $V_j^{+-}$ to $V_{j+1}^{--}$ and we have two sub-cases to consider. If some node $u \in V_j^{-+}$ receives label $j$ and degree $d_j$ in $G_n$, then $u$ is moved from $V_j^{-+}$ to $V_j^{--}$ ($V_{j+1}^{+-}$ and $V_{j+1}^{-+}$ each lose one node in this case); if some node $u \in V_j^{++}$ receives label $j$ and degree $d_j$ in $G_n$, then $u$ is moved from $V_j^{++}$ to $V_j^{+-}$ ($V_{j+1}^{+-}$ loses a node and gains a different node in this case). Thus, in any case, $|V_{j+1}^{+-}|-|V_{j+1}^{-+}|$ is unaffected by the process of assigning labels and degrees. Finally, we need to investigate what happens when communities are unlocked. Any node in a locked community at step $j$ is in $V_j^{+-}$ or $V_j^{--}$. Once a community is unlocked, all of the corresponding nodes in $V_j^{+-}$ move to $V_{j+1}^{++}$ and all of the corresponding nodes in $V_j^{--}$ move to $V_{j+1}^{-+}$. Thus, every node in a newly unlocked community decreases $V_{j+1}^{+-}$ by one or increases $V_{j+1}^{-+}$ by one, but not both. Therefore, $\Big(|V_j^{+-}|-|V_j^{-+}|\Big) - \Big(|V_{j+1}^{+-}|-|V_{j+1}^{-+}|\Big)$ is precisely the number of nodes in communities unlocked at step $j+1$. The claim now follows by induction. 

We have established that 
$$
p_j = \frac{|V_j^{++}||V_j^{-+}| + |V_j^{+-}||V_j^{-+}|}{|V_j^{++}||V_j^{+-}| + |V_j^{+-}||V_j^{-+}|} \leq 1. 
$$
Next, consider a node $v \in V_j^{-+}$. Then $v$ is given label $j$ and degree $d_j$ in $G_n$ if and only if some node $V_j^{+-}$ is chosen in $G_n^-$, the label is redirected to $V_j^{-+}$ in $G_n$, and $v$ is then chosen uniformly from the set $V_j^{-+}$ to receive the label in $G_n$. Thus, the probability that $v \in V_j^{-+}$ is assigned label $j$ and degree $d_j$ is 
\begin{align*}
\left(\frac{|V_j^{+-}|}{|V_j^{++}| + |V_j^{+-}|} \right) \left( \frac{|V_j^{++}||V_j^{-+}| + |V_j^{+-}||V_j^{-+}|}{|V_j^{++}||V_j^{+-}| + |V_j^{+-}||V_j^{-+}|} \right) \left( \frac{1}{|V_j^{-+}|} \right)
&=
\frac{1}{|V_j^{++}| + |V_j^{-+}|} \,.
\end{align*}
Consequently, a node $v$ in $V_j^{++}$ is labelled in $G_n$ at step $j$ with probability 
\begin{align*}
\left( 1 - \frac{|V_j^{-+}|}{|V_j^{++}| + |V_j^{-+}|} \right) \left( \frac{1}{|V_j^{++}|} \right)
&=
\left(\frac{|V_j^{++}|}{|V_j^{++}| + |V_j^{-+}|} \right) \left( \frac{1}{|V_j^{++}|} \right)
=
\frac{1}{|V_j^{++}| + |V_j^{-+}|} \,.
\end{align*}
Therefore, at every step $i \leq j \leq n$, the node chosen to receive label $j$ and degree $d_j$ is a uniform element of $V_j^{++} \cup V_j^{-+}$, the set of unlocked and unlabelled (that is, open) nodes in $G_n$ at step $j$. Lastly, the remaining part of the construction process of $G_n$ is equivalent to that of $\abcdDist$, and hence $G_n \sim \abcdDist$. 

We continue with the proof of claim (b). Let $C \in \textbf{C}_n$ satisfy $|C| = z$. Then the coupling ensures that $C$ is unlocked in both $G_n$ and $G_n^-$ before there is any deviation in the assignment process. Hence, if a node $v \in C$ receives label $j$ and degree $d_j$ in $G_n^-$, then $v$ will receive the same label in $G_n$ unless $v$ has already been labelled. If $v$ was already labelled in $G_n$ then this label is some $j' < j$. Since $d_1 \geq \dots \geq d_n$, $d_{j'} \geq d_j$. Therefore, the degree sequence $\mathbf{c}_z^-$ of $C$ in $G_n^-$ is bounded above point-wise by the degree sequence $\mathbf{c}_z$ in $G_n$. The proof now follows from Lemma~\ref{lem:deg dist lower}. 
\end{proof}

We continue with another modified version of phase 3 of the \textbf{ABCD} construction process. This new version will be used to prove the upper bound in Theorem~\ref{thm:main}. Fix $z$ with $s \leq z \leq n^\tau$ and define the construction process $\abcdUpper$, yielding a collection of degrees assigned to a collection of communities notated as $G_n^+$, as follows. 
\begin{enumerate}
\item Copy phases 1 and 2 of the \textbf{ABCD} construction process to get a degree distribution $\textbf{d}_n = (d_i, i \in [n])$ and a collection of communities $\textbf{C}_n = (C_j, j \in [L])$ each containing unassigned nodes. 
\item Copy phase~3 of the \textbf{ABCD} construction process until the communities of size $z$ are unlocked. This event occurs at step $i$ where $i$ is the smallest label satisfying $d_i \leq \frac{z-1}{1-\xi \phi}$. Let $n'$ be the number of locked nodes, i.e., the number of nodes in communities of size at most $z_i-1$ (recall that $z_i = \lceil d_i(1-\xi \phi) + 1 \rceil$). At this point, of the $n-n'$ unlocked nodes, we have assigned $i-1$ of them labels $1,\dots, i-1$ and corresponding degrees $d_1,\dots, d_{i-1}$ in some order. 
\item Now keep the communities of size at most $z_i-1$ locked and assign labels $i,\dots,n-n'$ and corresponding degrees $d_i,\dots,d_{n-n'}$ uniformly at random to the collection of unlocked and unassigned nodes. 
\item Finally, unlock the communities of size at most $z_i-1$ and assign the $n'$ unassigned nodes labels $n-n'+1,\dots,n$ and corresponding degrees $d_{n-n'+1},\dots,d_n$ in any order (we will later show that w.h.p.\ $d_{n-n'+1} = \dots = d_n = \delta$).
\end{enumerate}
Note that, by the end of step 3, all nodes in communities of size $z$ have been assigned a label and a degree. This labelling is all we need to complete the proof, and we include step 4 only for the sake of completeness.

We first show that a community $C_j$ in $G_n^+ \sim \abcdUpper$ with $z$ nodes has the desired degree distribution. Our statement this time is not as strong as Lemma~\ref{lem:deg dist lower}, though thanks to Lemma~\ref{lem:open communities} we can still stochastically bound the degree sequence of a community of size $z$ in $G_n^+$.
\begin{lemma}\label{lem:deg dist upper}
Let $G_n^+ \sim \abcdUpper$, let $C_j$ be a community in $G_n^+$ with $|C_j| = z$ and with degree sequence $\mathbf{c}_z^+$, and let $(X_i^+, 1 \leq i \leq z)$ be the i.i.d.\ sequence defined in Theorem~\ref{thm:main}. Then w.h.p.\ $\mathbf{c}_z^+$ is stochastically bounded above by $(X_i^+, 1 \leq i \leq z)$. 
\end{lemma}

\begin{proof}
As in the proof of Lemma~\ref{lem:deg dist lower}, we will use the Principle of Deferred Decisions, that is, at the beginning we only uncover some partial information about the degree sequence $\mathbf{d}_n$. As before, we first expose whether or not $d_i > \frac{z-1}{1-\xi \phi}$ and, if the inequality holds, then we expose the value of~$d_i$. However, if $d_i \le \frac{z-1}{1-\xi \phi}$, then we reveal $d_i$ only if $d_i = \delta$, and otherwise we do not expose additional information about $d_i$. 

By the construction of $G_n^+ \sim \abcdUpper$, we know that the sequence of degrees in $C_j$ is a uniform subsequence of $(d_i,\dots, d_{n-n'})$, where $i$ is the smallest labelled node satisfying $d_i \leq \frac{z-1}{1-\xi \phi}$ and $n'$ is the number of nodes in communities of size at most $z_i-1$. Then, letting $V'$ be as in Lemma~\ref{lem:open communities}, we have that $n' \leq |V'|$ and that  w.h.p.\ by Lemma~\ref{lem:open communities}, $|V'| < \omega n^{1-(\tau-\zeta)(2-\beta)}$ for any function $\omega = \omega(n) \to \infty$. Thus, w.h.p.\ $n' = o \left(n^{1-(\tau-\zeta)(2-\beta)/2}\right) = o(\epsilon n)$. (Recall that $\epsilon = n^{-(\tau-\zeta)(2-\beta)/2}$.) Since we aim for a statement that holds w.h.p., we may condition on this event.

Let $n''$ be the number of nodes of degree $\delta$. Note that $n''$ is simply a $\mathrm{Binomial(n-i,p_\delta)}$ random variable with
$$
p_\delta = \frac{\int_\delta^{\delta+1} x^{-\gamma} \, dx}{\int_\delta^{\Delta_z+1} x^{-\gamma} \, dx},
$$  
where $\Delta_z = \min \left\{ \frac{z-1}{1-\xi \phi} , n^\zeta \right\}$. 
It follows immediately from Chernoff's bound that w.h.p.\ we have 
$$
n'' = (n-i)p_\delta + \omega \sqrt{n} = (n-i)p_\delta + o(\epsilon n) = (n-i)p_\delta (1+o(\epsilon)), 
$$
the second equality holding since $1-(\tau-\zeta)(2-\beta)/2 > 1/2$. We may condition on this event too.

Let us now summarize our situation. The degree distribution of $C_j$ is a uniform subsequence of length $z$ of the sequence
\[
(d_i,\dots,d_{n-n'}) = (d_i,\dots,d_{n-n''}) ^\frown (d_{n-n''+1},\dots,d_{n-n'})
\] 
of $n-n'-i = (n-i)(1-o(\epsilon))$ degrees. ($\textbf{x} ^\frown \textbf{y}$ is the concatenation of sequences $\textbf{x}$ and~$\textbf{y}$.) The subsequence $(d_i,\dots,d_{n-n''})$ consists of degrees that are at least $\delta+1$ and at most $\Delta_z$; recall that, since we have not yet exposed these degrees, by Fact~\ref{fact:tpl condition} they are i.i.d.\ random variables with distribution $\tpl{\gamma}{\delta+1}{\Delta_z}$. On the other hand, $(d_{n-n''+1},\dots,d_{n-n'})$ is simply a sequence of $n''-n' = (n-i)p_\delta (1-o(\epsilon))$ copies of $\delta$.  

Now, let us provide a more careful argument to show that a uniform subsequence of $(d_i,\dots,d_{n-n''})$ of length $z$ satisfies the stochastic domination in the statement of the theorem. We sample $z$ times uniformly at random from this sequence (that may be viewed as a multi-set) without replacement and observe that each time we select $\delta$ with probability at least
\begin{align*}
\frac {n''-n'-z}{n-n'-i-z} = p_\delta (1-o(\epsilon)) 
&= \frac{(1-\epsilon-o(\epsilon^2))\int_\delta^{\delta+1} x^{-\gamma} \, dx}{(1-\epsilon)\int_\delta^{\delta+1} x^{-\gamma} \, dx + (1-\epsilon) \int_{\delta+1}^{\Delta_z+1} x^{-\gamma} \, dx} \\
&> \frac{(1-\epsilon)\int_\delta^{\delta+1} x^{-\gamma} \, dx}{(1-\epsilon)\int_\delta^{\delta+1} x^{-\gamma} \, dx + \int_{\delta+1}^{\Delta_z+1} x^{-\gamma} \, dx} \,.
\end{align*}
If we select a value other than $\delta$, then our selected degree has distribution $\tpl{\gamma}{\delta+1}{\Delta_z}$. Therefore, w.h.p.\ the random subsequence $\mathbf{c}_z^+$ is stochastically bounded from above by the i.i.d.\ sequence $(X_i^+, 1 \leq i \leq z)$ defined in Theorem~\ref{thm:main}, and the proof of the lemma is finished.
\end{proof}

We will now couple the constructions of $G_n \sim \abcdDist$ and $G_n^+ \sim \abcdUpper$ and prove the upper bound in Theorem~\ref{thm:main}. Contrast to the proof of the lower bound, we will first construct $G_n \sim \abcdDist$ and couple this construction with another construction $G_n^+$ which we will later show satisfies $G_n^+ \sim \abcdUpper$. 

\begin{proof}[Proof of Theorem~\ref{thm:main} (upper bound)]
Construct $G_n \sim \abcdDist$ with nodes labelled as $[n]$, degree sequence $\textbf{d}_n = (d_i, i \in [n])$, and community sequence $\textbf{C}_n = (C_j,j \in [L])$, and construct $G_n^+$ in parallel as follows. 
\begin{enumerate}
\item Let $G_n^+$ have degree sequence $\textbf{d}_n$ and community sequence $\textbf{C}_n$. 
\item Copy the degree assignment process of $G_n$ until the communities of size $z$ are unlocked. Let $i$ be the smallest labelled node satisfying $d_i \leq \frac{z-1}{1-\xi \phi}$ and let $n'$ be the number of nodes in communities of size at most $z_i-1$ (recall that $z_i = \lceil d_i(1-\xi \phi) + 1 \rceil$). Instead of unlocking communities progressively as we do in $G_n \sim \abcdDist$, we will keep the $n'$ nodes locked until we have assigned label $n-n'$ and degree $d_{n-n'}$ as we do in $\abcdUpper$. 
\item Now, for $j \in \{i,\dots ,n-n'\}$ starting with $j = i$, we first partition the nodes into three sets as follows. 
\begin{align*}
V_j^{++} &= \left\{ v : v \text{ is open in both } G_n \text{ and } G_n^+ \text{ at step } j \right\} \,, \\
V_j^{+-} &= \left\{ v : v \text{ is open in } G_n \text{ and closed in } G_n^+ \text{ at step } j \right\} \,, \\
V_j^{--} &= \left\{ v : v \text{ is closed in both } G_n \text{ and } G_n^+ \text{ at step } j \right\} \,.
\end{align*}
Note, distinct from the lower-bound, that $V_i^{+-} = \emptyset$, and that that there is no set $V_j^{-+}$. We need not define $V_j^{+-}$, as we will never encounter a scenario where a node is assigned in $G_n$ but unassigned in $G_n^+$. We now choose a node $v$ in $G_n$ to receive label $j$ and degree $d_j$ as per the $\abcdDist$ construction process (note that $v$ is chosen uniformly at random from $V_j^{++} \cup V_j^{+-}$). We then choose a node in $G_n^+$ to receive label $j$ and degree $d_j$ as follows.
\begin{itemize}
\item If $v \in V_j^{++}$, we give label $j$ and degree $d_j$ to $v$ in $G_n^+$.
\item If $v \in V_j^{+-}$, we give label $j$ and degree $d_j$ to a uniform node in $V_j^{++}$ in $G_n^+$. 
\end{itemize}
\item Finally, unlock the $n'$ locked nodes in $G_n^+$ and assign labels $n-n'+1, \dots , n$ and degrees $d_{n-n'+1},\dots,d_n$ uniformly among these newly unlocked nodes, independent of how these labels and degrees are assigned in $G_n$. 
\end{enumerate}
Similar to the previous coupling, the last step of the coupling is given only for the sake of completeness and has no bearing on the proof. We claim (a) that $G_n^+ \sim \abcdUpper$, and (b) that any community $C \in \textbf{C}_n$ of size $z$ with degree sequence $\mathbf{c}_z^+$ in $G_n^+$ and degree sequence $\mathbf{c}_z$ in $G_n$ satisfies $\mathbf{c}_z^+ \geq \mathbf{c}_z$ point-wise.  

Starting with claim (a), it is clear by the construction process $\abcdDist$ that $\textbf{d}_n$ and $\textbf{C}_n$ are valid sequences for $G_n^+ \sim \abcdUpper$. It is also clear that the degree assignment process in $G_n^+$ for nodes in communities of size at most $z_i-1$ is valid, since this assignment process is identical to that of $\abcdDist$ (which is identical to that of $\abcdUpper$ as well). We must then verify that, for $j  \in \{ i,\dots,n-n' \}$, the node in $G_n^+$ chosen to receive label $j$ and degree $d_j$ is a uniform node from the set of unassigned nodes in communities of size at least $z_i$. Note that this set of nodes is precisely $V_j^{++}$. For $u \in V_j^{++}$, $u$ is assigned label $j$ and degree $d_j$ in $G_n^+$ if $u$ is assigned this label and degree in $G_n$ or if a node $v \in V_j^{+-}$ is assigned this label and degree in $G_n$ and this label and degree is redirected to $u$ in $G_n^+$. Thus, the probability that $u \in V_j^{++}$ is labelled at step $j$ is 
\begin{align*}
\frac{1}{|V_j^{++}| + |V_j^{+-}|} + \left(\frac{|V_j^{+-}|}{|V_j^{++}| + |V_j^{+-}|}\right) \left( \frac{1}{|V_j^{++}|} \right) = \frac{1}{|V_j^{++}|}\,,
\end{align*}
and, in particular, the probability is equal for all $u \in V_j^{++}$. 
Therefore, at every step $i \leq j \leq n$, the node chosen to receive label $j$ and degree $d_j$ is a uniform element from the set of unlocked and unlabelled nodes in $G_n^+$ at step $j$, and this proves claim (a).

We continue with the proof of claim (b). Let $C \in \textbf{C}_n$ satisfy $|C| = z$. Then the coupling ensures that $C$ is unlocked in both $G_n^+$ and $G_n$ before there is any deviation in the assignment process. Hence, if a node $v \in C$ receives label $j$ and degree $d_j$ in $G_n$, then $v$ will receive the same label and degree in $G_n^+$ unless $v$ has already been given some label $j' < j$ and degree $d_{j'} \geq d_j$ in $G_n^+$. Therefore, the degree sequence $\mathbf{c}_z$ of $C$ in $G_n$ is bounded above point-wise by the degree sequence $\mathbf{c}_z^+$ in $G_n^+$. The proof now follows from Lemma~\ref{lem:deg dist upper}. 
\end{proof}

\section{Applications of Theorem~\ref{thm:main}}\label{sec:loops and multiedges}

In Section~\ref{sec:simulations} we provided empirical results supporting Corollary~\ref{cor:volumes} and Theorem~\ref{thm:loops and multiedges}, and showing that the asymptotic predictions based on these results reflect well the behaviour of the model even for moderately large values of $n$. In this section, we give the proofs of the theoretical, asymptotic, results. 

Let $X \sim \tpl{\gamma}{\delta}{\Delta}$ and recall that $\moment{\ell}{\gamma,\delta,\Delta} = \E{X^\ell}$. Unfortunately, there is no closed formula for $\moment{\ell}{\gamma,\delta,\Delta}$. However, in the coming proofs, we use the following standard technique to bound $\moment{\ell}{\gamma,\delta,\Delta}$ (and other related values) from above and below: 
\begin{align*}
\moment{\ell}{\gamma,\delta,\Delta} &= \sum_{k = \delta}^\Delta k^\ell \frac{\int_k^{k+1} x^{-\gamma} \, dx}{\int_\delta^{\Delta+1} x^{-\gamma} \, dx}
\leq \sum_{k = \delta}^\Delta \frac{\int_k^{k+1} x^{\ell-\gamma} \, dx}{\int_\delta^{\Delta+1} x^{-\gamma} \, dx}
= \frac{\int_\delta^{\Delta+1} x^{\ell-\gamma} \, dx}{\int_\delta^{\Delta+1} x^{-\gamma} \, dx} \,, \text{ and}\\
\moment{\ell}{\gamma,\delta,\Delta} &= \sum_{k =\delta}^\Delta k^\ell \frac{\int_k^{k+1} x^{-\gamma} \, dx}{\int_\delta^{\Delta+1} x^{-\gamma} \, dx}
\geq \sum_{k = \delta}^\Delta \left(\frac{k}{k+1}\right)^\ell \frac{\int_k^{k+1} x^{\ell-\gamma} \, dx}{\int_\delta^{\Delta+1} x^{-\gamma} \, dx}
\geq \left(\frac{\delta}{\delta+1}\right)^\ell \frac{\int_\delta^{\Delta+1} x^{\ell-\gamma} \, dx}{\int_\delta^{\Delta+1} x^{-\gamma} \, dx} \,.
\end{align*}

\subsection{Volumes of Communities (Proof of Corollary~\ref{cor:volumes})}

\begin{proof}[Proof of Corollary~\ref{cor:volumes}]
Let $G_n \sim \abcdDist$ with degree sequence $\mathbf{d}_n$, let $C_j$ be a community in $G_n$ with $|C_j| = z$, let $\mathbf{c}_j$ be the degree sequence of $C_j$, and let 
\[
\Delta_z = \min \left\{ \frac{z-1}{1-\xi \phi} , n^\zeta \right\}, \text{ where } \phi = 1 - \frac{1}{n^2} \sum_{j \in [L]} |C_j|^2 \,.
\]
Now let $(X_i^-, 1 \leq i \leq z)$ and $(X_i^+, 1 \leq i \leq z)$ be as in Theorem~\ref{thm:main}. Then, since w.h.p.\ $\mathbf{c}_j$ is stochastically dominated from below by $(X_i^-, 1 \leq i \leq z)$, we get that w.h.p.\
\begin{align*}
\frac{\E{\mathrm{vol}(C_j)}}{z} 
&\geq \frac{1}{z} \, \E{\sum_{i = 1}^z X_i^-}
= \moment{1}{\gamma,\delta,\Delta_z} \,,
\end{align*}
and since w.h.p.\ $\mathbf{c}_j$ is stochastically dominated from above by $(X_i^+, 1 \leq i \leq z)$, we get that w.h.p.\
\begin{align*}
\frac{\E{\mathrm{vol}(C_j)}}{z} 
&\leq \frac{1}{z} \, \E{\sum_{i = 1}^z X_i^+}
= (1+o(1)) \frac{1}{z} \, \E{\sum_{i = 1}^z X_i^-}
= (1+o(1)) \moment{1}{\gamma,\delta,\Delta_z} \,,
\end{align*}
which establishes the first claim in Corollary~\ref{cor:volumes}. Next, we have 
\begin{align*}
\moment{1}{\gamma,\delta,n^\zeta} - \moment{1}{\gamma,\delta,\Delta_z}
&=
\left( \sum_{k=\delta}^{n^\zeta} k \frac{\int_k^{k+1} x^{-\gamma} \, dx}{\int_\delta^{n^\zeta+1} x^{-\gamma} \, dx} - \sum_{k=\delta}^{\Delta_z} k \frac{\int_k^{k+1} x^{-\gamma} \, dx}{\int_\delta^{\Delta_z+1} x^{-\gamma} \, dx} \right)\\
&=
\left( 1+ O(\Delta_z^{1-\gamma}) \right) \left( \sum_{k=\delta}^{n^\zeta} k \frac{\int_k^{k+1} x^{-\gamma} \, dx}{\int_\delta^{n^\zeta+1} x^{-\gamma} \, dx} - \sum_{k=\delta}^{\Delta_z} k \frac{\int_k^{k+1} x^{-\gamma} \, dx}{\int_\delta^{n^\zeta+1} x^{-\gamma} \, dx} \right)\\
&=
\left( 1+ O(\Delta_z^{1-\gamma}) \right) \sum_{k=\Delta_z+1}^{n^\zeta} k \frac{\int_k^{k+1} x^{-\gamma} \, dx}{\int_\delta^{n^\zeta+1} x^{-\gamma} \, dx}\\
&\leq
\left( 1+ O(\Delta_z^{1-\gamma}) \right) \frac{\int_{\Delta_z+1}^{n^\zeta+1} x^{1-\gamma} \, dx}{\int_\delta^{n^\zeta+1} x^{-\gamma} \, dx}\\
&=
O(\Delta_z^{2-\gamma}) \,.
\end{align*}
The second claim in Corollary~\ref{cor:volumes} now follows since w.h.p.\
\begin{align*}
\frac{\E{\mathrm{vol}(C_j)}}{z} 
&= (1+o(1)) \Big( \moment{1}{\gamma,\delta,n^\zeta} - \left( \moment{1}{\gamma,\delta,n^\zeta} - \moment{1}{\gamma,\delta,\Delta_z} \right) \Big)\\
&= (1+o(1)) \Big( \moment{1}{\gamma,\delta,n^\zeta} - O(\Delta_z^{2-\gamma}) \Big) 
\end{align*}
and, since $\Delta_z = \Theta(\min\{z,n^{\zeta}\})$, we have that $\Delta_z \to \infty$ as $z \to \infty$. 
\end{proof}

\subsection{Loops and Multi-edges (Proof of Theorem~\ref{thm:loops and multiedges})}

Throughout this section, it will be useful to refer to the multi-graph generated by the first four phases of the \textbf{ABCD} construction. Write $G_n \sim \abcdFour$ to mean $G_n$ is the hypergraph generated by the first four phases.

Before tackling the upper-bounds in Theorem~\ref{thm:loops and multiedges}, we first prove that the number of loops and multi-edges in $G_n \sim \abcdFour$ is asymptotically bounded from below by the number of communities. In fact, we show that the number of loops in community graphs alone is asymptotically bounded in this way. 

\begin{lemma}\label{lem:loops and multiedges lowerbound}
Let $G_n \sim \abcdFour$ with $L$ communities and let $S_c$ be the number of loops in community graphs in $G_n$. Then w.h.p.\
\[
S_c = \Omega(L) \,.
\]
\end{lemma}
\begin{proof}
Fix a constant $z$ large enough so that $z \geq s$ and $\lfloor (1-\xi) \Delta_z \rfloor \geq 2$ and let $G_{n,j}$ be a community graph in $G_n$ with $|C_j| = z$ and with degree sequence $(Y_i, i \in C_j)$ (recall that $Y_i = \round{(1-\xi)d_i}$ where $\round{\cdot}$ is a random rounding function). Then, by the lower bound in Theorem~\ref{thm:main}, a uniformly random degree $Y_i$ is stochastically bounded from below by $\lfloor (1-\xi) X \rfloor$ where $X \sim \tpl{\gamma}{\delta}{\Delta_z}$. Thus, by the stochastic bound, we have
\[
\p{Y_i = \left\lfloor (1-\xi)\Delta_z \right\rfloor} 
\geq 
\p{X = \Delta_z}
> 0 \,.
\]
Thus, w.h.p.\ a linear proportion of community graphs with $z$ nodes contain at least one node $v$ with $\mathrm{deg}(v) = \left\lfloor (1-\xi)\Delta_z \right\rfloor \geq 2$. Furthermore, a node with this degree generates a loop in $G_n \sim \abcdFour$ with positive probability, and so w.h.p.\ a linear proportion of community graphs with $z$ nodes contain at least one loop. Finally, as the number of communities of size $z$ is w.h.p.\ $\Theta(L)$, the lemma follows. 
\end{proof}

We continue now with the upper-bounds. The heart of Theorem~\ref{thm:loops and multiedges} is the following lemma. 

\begin{lemma} \label{lem:community loops and multi-edges}
Fix $z > \Delta > \delta > 0$ and $\gamma \in (2,3)$. Let $\mathbf{q}_z = (q_i,i \in [z])$ be a sequence of i.i.d. random variables with $q_i \sim \tpl{\gamma}{\delta}{\Delta}$ and let $H_z$ be sampled as the configuration model with degree sequence $\mathbf{q}_z$. Let $S$ and $M$ be the number of self-loops and, respectively, multi-edges in $H_z$. Then
\begin{align*}
\E{S} &\leq (1+O(\Delta^{\gamma-3})) \, c(\gamma,\delta) \Delta^{3-\gamma} \,, \text{ and}\\
\E{M} &\leq (1+O(\Delta^{\gamma-3})) \, c(\gamma,\delta)^2 \Delta^{6-2\gamma} \,,
\end{align*}
where 
\[
c(\gamma,\delta) = \frac{(\gamma-1)\delta^{\gamma-2}}{2(3-\gamma)} \,.
\]
\end{lemma}

\begin{proof}
Let us recall equations~(\ref{eq:S bound}) and~(\ref{eq:M bound}) for ease of reading:
\[
\Cexp{S}{\mathbf{q}_z} = \frac{\sum_{i \in [z]} q_i(q_i-1)}{2 \left(\sum_{i=1}^z q_i - 1\right)} 
\leq 
\frac{1}{2} \frac{\sum_{i \in [z]} q_i^2}{\sum_{i=1}^z q_i - 1} 
\,,
\]
and
\[
\Cexp{M}{\mathbf{q}_z} \leq \frac{\sum_{1 \leq i < j \leq z} q_i(q_i-1)q_j(q_j-1)}{2\left( \sum_{i=1}^z q_i - 1 \right)\left( \sum_{i=1}^z q_i - 3 \right)} 
\leq 
\frac{1}{2} \frac{\sum_{1 \leq i < j \leq z} q_i^2 q_j^2}{\left( \sum_{i=1}^z q_i - 3 \right)^2} 
\,.
\]

Next, for independent $X,Y \sim \tpl{\gamma}{\delta}{\Delta}$ we have 
\begin{align*}
\E{X^2} &= \sum_{k=\delta}^\Delta \frac{k^2 \int_k^{k+1} x^{-\gamma} \, dx}{\int_\delta^{\Delta +1} x^{-\gamma} \, dx} \\
&\leq \frac{\int_\delta^{\Delta +1} x^{2-\gamma} \, dx}{\int_{\delta}^{\Delta+1} x^{-\gamma} \, dx} \\
&= \left(\frac{\gamma-1}{3-\gamma}\right) \left(\frac{(\Delta+1)^{3-\gamma} - \delta^{3-\gamma}}{\delta^{1-\gamma} - (\Delta +1)^{1-\gamma}} \right)  \\
&= (1+O(\Delta^{\gamma-3} + \Delta^{1-\gamma})) \, \left(\frac{\gamma-1}{3-\gamma}\right) \, \delta^{\gamma - 1} \Delta^{3-\gamma}\\
&= (1+O(\Delta^{\gamma-3})) \, \left(\frac{\gamma-1}{3-\gamma}\right) \, \delta^{\gamma - 1} \Delta^{3-\gamma} \,,
\end{align*}
and 
\begin{align*}
\E{X^2Y^2} &= \E{X^2}\E{Y^2}\\
&= (1+O(\Delta^{\gamma-3})) \, \left(\frac{\gamma-1}{3-\gamma}\right)^2 \, \delta^{2\gamma - 2} \Delta^{6-2\gamma} \,.
\end{align*}
Now, since $\mathbf{q}_z$ contains i.i.d.\ random variables, and since $\sum_{i=1}^z q_i \geq \delta z$, it follows from (\ref{eq:S bound}) that 
\begin{align*}
\E{S} &= \E{\Cexp{S}{\mathbf{q}_z}}\\
&\leq
\frac{1}{2} \, \E{\frac{\sum_{i \in [z]} q_i^2}{\sum_{i \in [z]} q_i - 1}}\\
&\le
\frac{1}{2(\delta z - 1)} \sum_{i \in [z]}\E{q_i^2}\\
&\leq
(1+O(\Delta^{\gamma-3})) \left(\frac{1}{2 \delta z} \right) \left( z \left(\frac{\gamma-1}{3-\gamma}\right) \delta^{\gamma-1} \Delta^{3-\gamma} \right) \\
&=
(1+O(\Delta^{\gamma-3})) \left(\frac{(\gamma-1)\delta^{\gamma-2}}{2(3-\gamma)}\right) \Delta^{3-\gamma}\,,
\end{align*}
and from (\ref{eq:M bound}) that 
\begin{align*}
\E{M}
&=
\E{\Cexp{M}{\mathbf{q}_z}}\\
&\leq
\frac{1}{2} \, \E{ \frac{\sum_{1 \leq i < j \leq z} q_i^2 q_j^2}{\left( \sum_{i=1}^z q_i - 3 \right)^2} }\\
&\leq
\frac{1}{2(\delta z - 3)^2} \sum_{1 \leq i < j \leq z} \E{q_i^2 q_j^2}\\
&\leq 
(1+O(\Delta^{\gamma-3})) \left(\frac{1}{2\delta^2 z^2}\right) {z \choose 2}  \, \left(\frac{\gamma-1}{3-\gamma}\right)^2 \, \delta^{2\gamma - 2} \Delta^{6-2\gamma}\\
&\leq 
(1+O(\Delta^{\gamma-3})) \left(\frac{(\gamma-1)\delta^{\gamma-2}}{2(3-\gamma)}\right)^2 \Delta^{6-2\gamma} \,.
\end{align*}
Note that, in the first computation, we use the fact that 
\[
\frac{1}{2(\delta z - 1)} = (1+O(z^{-1})) \frac{1}{2 \delta z} = (1+O(\Delta_z^{\gamma-3})) \frac{1}{2 \delta z} \,,
\]
and in the second computation, we use the fact that 
\[
\frac{1}{2(\delta z - 3)^2} = (1+O(z^{-1})) \frac{1}{2 \delta^2 z^2} = (1+O(\Delta_z^{\gamma-3})) \frac{1}{2 \delta^2 z^2} \,.
\]
This finishes the proof of the lemma.
\end{proof}

We are now ready to prove Theorem~\ref{thm:loops and multiedges}.

\begin{proof}[Proof of Theorem~\ref{thm:loops and multiedges}]
Let $G_n \sim \abcdFour$ with degree sequence $\textbf{d}_n = (d_i, i \in [n])$, and let $S_c, M_c, S_b, M_b$ and $M_{bc}$ be as in the statement of the theorem. Starting with $S_b$ and $M_b$, note that the degree sequence in $G_{n,0}$ is $(Z_i, i \in [n])$ where $Z_i = \round{\xi d_i}$. Thus, $Z_i \leq d_i$, meaning by Lemma~\ref{lem:community loops and multi-edges} that
\begin{align*}
\E{S_b} &\leq \left(1 + O(n^{\zeta(\gamma-3)}) \right) c(\gamma,\delta) \big(n^\zeta \big)^{3-\gamma} = O(n^{\zeta(3-\gamma)}) \,, \text{ and}\\
\E{M_b} &\leq \left(1 + O(n^{\zeta(\gamma-3)}) \right) c(\gamma,\delta)^2 \big(n^\zeta\big)^{6-2\gamma} = O(n^{\zeta(6-2\gamma)}) \,,
\end{align*}
proving claims 3.\ and 4. 

Continuing with $S_c$ and $M_c$, for community graph $G_{n,j}$ with $|C_j| = z$ let $S_{c,j}$ and $M_{c,j}$ be the number of loops and multi-edges in $G_{n,j}$. Note that, for any node $i \in C_j$, the degree of $i$ in $G_{n,j}$ is $Y_i \leq d_i$. Thus, by Theorem~\ref{thm:main}, $Y_i$ is stochastically bounded from above by the random variable $Y \sim \tpl{\gamma}{\delta+1}{\Delta_z}$. Then, again by Lemma~\ref{lem:community loops and multi-edges}, we have that 
\begin{align*}
\Cexp{S_{c,j}}{|C_j|=z} &\leq \left(1 + O\left( \Delta_z^{\gamma-3} \right) \right) c(\gamma,\delta+1) \Delta_z^{3-\gamma} \,, \text{ and}\\
\Cexp{M_{c,j}}{|C_j|=z} &\leq \left(1 + O(\Delta_z^{\gamma-3}) \right) c(\gamma,\delta+1)^2 \Delta_z^{6-2\gamma} \,.
\end{align*}
For the remainder of the proof, we write $c = c(\gamma,\delta+1)$ to simplify notation. Recall from phase 2 of the construction process of $G_n$ that $|C_j| \sim \tpl{\beta}{s}{n^\tau}$. Therefore, 
\begin{align*}
\E{S_{c,j}} 
&=
\sum_{z = s}^{n^\tau} \Cexp{S_{c,j}}{|C_j| = z} \p{|C_j| = z}\\
&\leq
\sum_{z = s}^{n^\tau} \left(1 + O\left( \Delta_z^{\gamma-3} \right) \right) c \, \Delta_z^{3-\gamma} \frac{\int_z^{z+1} y^{-\beta} \, dy}{\int_s^{n^\tau + 1} y^{-\beta} \, dy} \,.
\end{align*}
We split the sum at the community size $z^*$, where $z^*$ is minimal with the property that
\[
\left\lfloor \frac{z^*-1}{1-\xi\phi} \right\rfloor \geq n^\zeta \,.
\]
Note that, for $z \leq z^*$, $\Delta_z = \Theta(z)$, and for $z \geq z^*$, $\Delta_z = n^\zeta$. Let $c'$ be a constant satisfying $\Delta_z^{3-\gamma} \leq c' z^{3-\gamma}$ for all $s \leq z \leq z^*$. For the first part of the sum, we have
\begin{align*}
&\sum_{z = s}^{z^*} \left(1 + O\left( \Delta_z^{\gamma-3} \right) \right) c \, \Delta_z^{3-\gamma} \frac{\int_z^{z+1} y^{-\beta} \, dy}{\int_s^{n^\tau + 1} y^{-\beta} \, dy}\\
\geq \ &\sum_{z = s}^{z^*} \left(1 + O\left( z^{\gamma-3} \right) \right) c c' z^{3-\gamma} \frac{\int_z^{z+1} y^{-\beta} \, dy}{\int_s^{n^\tau + 1} y^{-\beta} \, dy}\\
\leq \ & 
\left(1 + O\left( s^{\gamma-3} \right) \right) cc' \sum_{z = s}^{z^*} \frac{\int_z^{z+1} y^{3-\gamma-\beta} \, dy}{\int_s^{n^\tau + 1} y^{-\beta} \, dy}\\
= \ & 
\left(1 + O\left( s^{\gamma-3} \right) \right) cc' \frac{\int_s^{z^*+1} y^{3-\gamma-\beta} \, dy}{\int_s^{n^\tau + 1} y^{-\beta} \, dy}\\
= \ & 
\left(1 + O\left( s^{\gamma-3} \right) \right) cc' \left(\beta-1\right) s^{1-\beta} \left( \frac{\left(z^*+1\right)^{4-\gamma-\beta} - s^{4-\gamma-\beta}}{4-\gamma-\beta} \right) \\
= \ &
O\left(1 + (z^*)^{4-\gamma-\beta} \right) \\
= \ &
O\left(1 + n^{\zeta(4-\gamma-\beta)} \right) \,.
\end{align*}
For the second part of the sum, we have
\begin{align*}
& \sum_{z = z^*+1}^{n^\tau} \left(1 + O\left( \Delta_z^{\gamma-3} \right) \right) c \, \Delta_z^{3-\gamma} \frac{\int_z^{z+1} y^{-\beta} \, dy}{\int_s^{n^\tau + 1} y^{-\beta} \, dy}\\
= \ & \sum_{z = z^*+1}^{n^\tau} \left(1 + O\left( n^{\zeta(\gamma-3)} \right) \right) c \, n^{\zeta(3-\gamma)} \frac{\int_z^{z+1} y^{-\beta} \, dy}{\int_s^{n^\tau + 1} y^{-\beta} \, dy}\\
= \ &
\left(1 + O\left( n^{\zeta(\gamma-3)} \right) \right) c n^{\zeta(3-\gamma)} \sum_{z = z^*+1}^{n^\tau}  \frac{\int_z^{z+1} y^{-\beta} \, dy}{\int_s^{n^\tau + 1} y^{-\beta} \, dy}\\
= \ &
\left(1 + O\left( n^{\zeta(\gamma-3)} \right) \right) c n^{\zeta(3-\gamma)}  \frac{\int_{z^*+1}^{n^\tau+1} y^{-\beta} \, dy}{\int_s^{n^\tau + 1} y^{-\beta} \, dy}\\
= \ &
\left(1 + O\left( n^{\zeta(\gamma-3)} \right) \right) c n^{\zeta(3-\gamma)}  \frac{(z^*+1)^{1-\beta} - (n^\tau+1)^{1-\beta}}{s^{1-\beta} - (n^\tau+1)^{1-\beta}}\\
= \ & 
O\left( n^{\zeta(3-\gamma)} (z^*)^{1-\beta} \right)\\
= \ & 
O\left( n^{\zeta(3-\gamma)} n^{\zeta(1-\beta)} \right)\\
= \ & 
O\left( n^{\zeta(4-\gamma-\beta)} \right) \,,
\end{align*}
and thus, $\E{S_{c,j}} = O(1+n^{\zeta(4-\gamma-\beta)})$. An analogous calculation shows that $\E{M_{c,j}} = O(1+n^{\zeta(7-2\gamma-\beta)})$. Claims 1.\ and 2.\ now follow from linearity of expectation, along with the fact that w.e.p.\ the number of communities in $G_n$ is $\Theta(n^{1-\tau(2-\beta)})$. 

Claim 5.\ states that $\E{M_{bc}} = o(M_c)$. To see this, let $C_j$ be a community in $G_n$ and let $u,v \in C_j$. Now let $M_c(u,v)$ be the number of $\{u,v\}$ multi-edge pairs in $G_{n,j}$, and let $M_{bc}$ be the number of $\{u,v\}$ multi-edge pairs with one edge in $G_{n,j}$ and the other in $G_{n,0}$. Then 
\[
\Cexp{M_c(u,v)}{\textbf{d}_n} = \Theta\left( \frac{d_u^2 d_v^2}{\left(\sum_{i \in C_j} d_i \right)^2} \right) \,,
\]
whereas
\[
\Cexp{M_{bc}(u,v)}{\textbf{d}_n} = \Theta\left( \frac{d_u^2 d_v^2}{\left(\sum_{i \in C_j} d_i \right) \left(\sum_{i \in [n]} d_i \right)} \right) \,.
\]
Since $\sum_{i \in C_j} d_i = o\left( \sum_{i \in [n]} d_i \right)$ for all communities $C_j$, we get that $\E{M_{bc}(u,v)} = o\left( \E{M_{c}(u,v)} \right)$, and Claim 5.\ follows from linearity of expectation. 

Finally, we know that w.e.p.\ $L = \Theta(n^{1-\tau(2-\beta)})$ and that w.h.p.\ $\E{S_c} = \Omega(L)$. Now suppose that $\gamma+\beta > 4$ and that $2\zeta(3-\gamma) + \tau(2-\beta) \leq 1$, and note that these two inequalities imply that $2\gamma+\beta > 3+\gamma+\beta > 7$ and that $3-\gamma + \tau(2-\beta) \leq 1$. Therefore, under this assumption, w.h.p.\ we have
\begin{align*}
\E{S_c}
&=
O \left( (n^{1-\tau(2-\beta)})(1+n^{\zeta(4-\gamma-\beta)}) \right)
=
O \left(n^{1-\tau(2-\beta)}\right) \,, \\
\E{M_c + M_{bc}}
&=
(1+o(1)) \E{M_c}
=
O \left( (n^{1-\tau(2-\beta)})(1+n^{\zeta(7-2\gamma-\beta)}) \right)
=
O \left(n^{1-\tau(2-\beta)}\right) \,, \\
\E{S_b}
&=
O \left( n^{\zeta(3-\gamma)} \right)
=
O \left( n^{1-\tau(2-\beta)} \right) \,, \text{ and} \\
\E{M_b}
&=
O \left( n^{2(\zeta(3-\gamma))} \right)
=
O \left( n^{1-\tau(2-\beta)} \right) \,,
\end{align*}
which proves the final claim. 
\end{proof}

\section{Conclusion}\label{sec:conclusions}

The \textbf{ABCD} model generates a random graph with ground truth community structure and power-law distributions for both the degrees and the community sizes and, unlike the \textbf{LFR} model, allows for deep theoretical study to understand its properties. We take advantage of this feature and show via Theorem~\ref{thm:main} that \textbf{ABCD} graphs exhibit an interesting self-similar property, namely, that the degree distribution of individual communities is asymptotically the same as the appropriately normalized degree distribution of the entire graph. We argue that this result is a powerful tool for studying the structure of \textbf{ABCD} graphs and justify the argument by providing two consequences of Theorem~\ref{thm:main}, namely, Corollary~\ref{cor:volumes} which gives an improved understanding of community graph volumes, and Theorem~\ref{thm:loops and multiedges} which gives a new understanding of the number of self-loops and multi-edges that are generated during phase~4 of the construction process. 
This information is useful for practitioners since (a) rewiring self-loops and multi-edges to keep the graph simple is an expensive part of the algorithm, and (b) every rewiring causes the underlying configuration models to deviate slightly from uniform simple graphs on their corresponding degree sequences.
We complement the theoretical work with three experiments which, in all cases, show that the three main results given in this paper reveal themselves even for \textbf{ABCD} graphs with a relatively small number of nodes.


\bigskip

Let us finish with some open problems. We have shown two examples of how Theorem~\ref{thm:main} can help us understand the nature of \textbf{ABCD} graphs. There are more applications of Theorem~\ref{thm:main} that we do not explore here. Essentially, any result that holds for a configuration model on an i.i.d.\ degree sequence, sampled as $\tpl{\gamma}{\delta}{\Delta}$ for some $\gamma \in (2,3)$, should hold for a community graph in $G_n \sim \abcdFour$. With additional work, it may also be true that such results hold for a community graph in $G_n \sim \abcdDist$. Possible avenues for $G_n \sim \abcdDist$ include studying its diameter, its diffusion rate, its clustering coefficient, etc.

\bigskip

Our results in Corollary~\ref{cor:volumes} and Theorem~\ref{thm:loops and multiedges} are results only in expectation, though our experiments indicate that the behaviour of at least community volumes is quite tight. Given that the truncated power-law $\tpl{\gamma}{\delta}{n^\zeta}$ has unbounded second moment, and that $\tpl{\beta}{s}{n^\tau}$ has unbounded first moment, any study involving concentration will prove to be challenging. However, considering that the collection of community degree sequences partition the degree sequence of the whole graph, it is possible that these sequences exhibit self-correcting behaviour, and this is a potential road-map to a tighter version of our results. 

\bigskip

In Theorem~\ref{thm:loops and multiedges} we only show that collisions are bounded below asymptotically by $\Omega(L)$, and our justification is very liberal. On the other hand, our experimental results suggest that the number of collisions is, in fact, $\omega(L)$ when $\gamma + \beta \leq 4$ or when $2\zeta(3-\gamma) + \tau(2-\beta) > 1$. Thus, there is potential room to improve Theorem~\ref{thm:loops and multiedges} by tightening the lower-bound. 

\bibliography{ref}

\end{document}